\documentclass[11pt]{article}

\usepackage{amsmath,amssymb,amsthm} 
\usepackage{graphicx}              
\usepackage{hyperref}              
\usepackage{geometry}              
\usepackage{balance} 
\usepackage{amsmath}
\usepackage{amsthm}
\usepackage{subcaption}
\usepackage{booktabs}
\usepackage{xcolor}
\usepackage{xspace}
\usepackage{tikz}
\usepackage[T1]{fontenc}
\usepackage[ruled, vlined, linesnumbered]{algorithm2e}
\usepackage[numbers]{natbib}
\newtheorem{example}{Example}
\newtheorem{theorem}{Theorem}
\newtheorem{proposition}{Proposition}
\newtheorem{corollary}{Corollary}
\theoremstyle{definition}
\newtheorem{definition}{Definition}

\newcommand{\MDJR}{{\tt MDJR}\xspace}
\newcommand{\MDEJR}{{\tt MDEJR}\xspace}
\newcommand{\MREJR}{{\tt MREJR}\xspace}
\newcommand{\NEWSAT}{sparse-SAT\xspace}
\title{The Degree of (Extended) Justified Representation and Its Optimization}
\author{%
  Biaoshuai Tao \\
  John Hopcroft Center for Computer Science\\
  Shanghai Jiao Tong University \\
  \texttt{bstao@sjtu.edu.cn} \\
  \and
  Chengkai Zhang \\
  Shanghai Jiao Tong University \\
  \texttt{zhangchengkai@sjtu.edu.cn} \\
  \and
  Houyu Zhou \\
  UNSW Sydney \\
  \texttt{houyu.zhou@unsw.edu.au} \\
}
\date{}

\begin{document}

\maketitle 

\begin{abstract}
Justified Representation (JR)/Extended Justified Representation (EJR) is a desirable axiom in multiwinner approval voting. In contrast to that (E)JR only requires at least \emph{one} voter to be represented in every cohesive group, we study its optimization version that maximizes the \emph{number} of represented voters in each group. Given an instance, we say a winning committee provides a JR degree (EJR degree, resp.) of $c$ if at least $c$ voters in each $\ell$-cohesive group ($1$-cohesive group, resp.) have approved $\ell$ ($1$, resp.) winning candidates. Hence, every (E)JR committee provides the (E)JR degree of at least $1$. Besides proposing this new property, we propose the optimization problem of finding a winning committee that achieves the maximum possible (E)JR degree, called \MDJR and \MDEJR, corresponding to JR and EJR respectively.

We study the computational complexity and approximability of \MDJR of \MDEJR. An (E)JR committee, which can be found in polynomial time, straightforwardly gives a $(k/n)$-approximation. We also show that the original algorithms for finding a JR and an EJR winner committee are also $1/k$ and $1/(k+1)$ approximation algorithms for \MDJR and \MDEJR respectively. On the other hand, we show that it is NP-hard to approximate \MDJR and \MDEJR to within a factor of $\left(k/n\right)^{1-\epsilon}$ and to within a factor of $(1/k)^{1-\varepsilon}$, for any $\epsilon>0$, which complements the positive results. Next, we study the fixed-parameter-tractability of this problem. We show that both problems are W[2]-hard if $k$, the size of the winning committee, is specified as the parameter. However, when $c_{\text{max}}$, the maximum value of $c$ such that a committee that provides an (E)JR degree of $c$ exists, is additionally given as a parameter, we show that both \MDJR and \MDEJR are fixed-parameter-tractable.
\end{abstract}

\newcommand{\BibTeX}{\rm B\kern-.05em{\sc i\kern-.025em b}\kern-.08em\TeX}


\section{Introduction}
An approval-based committee voting rule (ABC rule) plays a crucial role in collective decision-making by determining a committee from a set of $m$ candidates $C=\{c_1,\dots,c_m\}$, a set of $n$ voters $N=\{1,\dots,n\}$ who each approve a subset $A_i$ of $C$, and an integer $k$ representing the desired committee size. Let the list $\mathbf{A}=(A_1,\dots,A_n)$ of approval ballots be the ballot profile. Formally, ABC rules take a tuple $(N,C,\mathbf{A},k)$ as input, where $k$ is a positive integer that satisfies $k\le |C|$, and return one or more 
size-$k$ subsets $W\subseteq C$, which are called the winning committees. In many concrete voting systems a tiebreaking method is included so that a resolute outcome is guaranteed. In some cases, $N$ and $C$ are omitted from the notation when they are clear from the context. ABC rules are widely applied in various contexts, including the election of representative bodies (such as supervisory boards and trade unions), identifying responses to database queries \cite{DBLP:conf/www/DworkKNS01,DBLP:conf/ijcai/SkowronLBPE17,DBLP:conf/ijcai/IsraelB21}, selecting validators in consensus protocols like blockchain \cite{DBLP:conf/aft/CevallosS21}, making collective recommendations for groups \cite{DBLP:conf/ijcai/LuB11,DBLP:conf/aaai/LuB15}, and facilitating discussions on proposals within liquid democracy \cite{behrens2014principles}. Additionally, as committee elections fall under the domain of participatory budgeting (PB) \cite{cabannes2004participatory,goel2016knapsack}, a strong grasp of ABC rules is indispensable for designing effective PB methods.

The applicability of various ABC rules often depends on the particular context, yet an important requirement one often imposes on an ABC rule is that it can accurately reflect the voters' preferences, e.g., every large group of voters should justify a seat in the committee. In recent years, a desirable axiom has been proposed and developed, \emph{Justified Representation} (JR), which requires every group of at least $n/k$ voters that have at least one common candidate should be represented.

Given a profile $\mathbf{A}$ of $n$ approval preferences, a subgroup of voters, denoted by $N'\subseteq N$ is called an $\ell$-cohesive group for some $\ell \in \mathbb{N}$, if $|N'| \ge \ell \cdot \frac{n}{k}$ and $|\bigcap_{i\in N'} A_i| \ge \ell$. If $\ell=1$, we say it is a cohesive group for short.

\begin{definition}[Justified representation (JR)]\emph{\cite{DBLP:journals/scw/AzizBCEFW17}}
    Given a ballot profile $\mathbf{A}=(A_1,\dots, A_n)$ over a candidate set $C$ and a committee size $k$, we say that a set of candidates $W$ of size $|W| = k$ satisfies JR for $(\mathbf{A}, k)$ if, for every cohesive group (defined right above), there is at least one voter which approves at least one candidate in $W$.

    We say that an ABC rule satisfies JR if for each profile $\mathbf{A}$ and committee size $k$, each winning committee provide JR.
\end{definition}

Many well-known ABC rules have been shown to satisfy JR, from a simple greedy approval voting rule to a more complex proportional approval voting (PAV) rule. Hence, some stronger JR axioms are proposed to distinguish the existing approval voting rules, such as \emph{Extended Justified Representation} (EJR). 
EJR requires that in every $\ell$-cohesive group, at least one member is represented. Hence, EJR implies JR.

\begin{definition}[Extended Justified Representation]\emph{\cite{DBLP:journals/scw/AzizBCEFW17}}
\label{def:EJR}
    Given a ballot profile $\mathbf{A}=(A_1,\dots, A_n)$ over a candidate set $C$ and a committee size $k$, we say that a set of candidates $W$ of size $|W| = k$ satisfies EJR for $(\mathbf{A}, k)$ if, for every $\ell$-cohesive group with every $\ell\in[k]$, there is at least one voter which approves at least $\ell$ candidates in $W$.

    We say that an ABC rule satisfies EJR if for each profile $\mathbf{A}$ and committee size $k$, each winning committee provide EJR.
    
\end{definition}

While there are a wealth of concepts associated with JR, there are instances where comparing the performance of different winning committees becomes impossible. To illustrate this situation, consider a brief example.

\begin{example} \label{eg: tiny}
    Let $k=1$, $N=\{1,2,3,4\}$, $C=\{c_1,c_2,c_3,c_4\}$, $A_1=\{c_1\}$, and $A_i=A_{i-1}\cup\{c_i\}$ for every $i=2,3,4$. We are interested in four different committees $W_i=\{c_i\}$ for every $i\in [4]$.
\end{example}

Since Example \ref{eg: tiny} is a single winner voting instance,  we have $k=\ell=1$. Clearly, $N$ is the only cohesive group in this example, and all four committees in Example \ref{eg: tiny} satisfy all the above JR axioms: besides the JR and EJR axioms mentioned earlier, one can verify that those committees also satisfy some other JR axioms such as FJR \cite{DBLP:conf/nips/PierczynskiSP21}, PJR+, and EJR+ \cite{DBLP:conf/sigecom/Brill023}. However, $W_{i-1}$ is intuitively better than $W_i$ since it satisfies more voters. 
To distinguish the performance of these committees, besides proposing new JR variants, another approach to addressing this scenario involves leveraging \emph{quantitative} techniques to model a hierarchy of JR, which has a similar spirit to the work of using quantitative techniques to model proportionality. 

\begin{definition}[Proportionality degree] \label{def:prop-degree}
\cite{DBLP:conf/sigecom/Skowron21}
    Fix a function $f: \mathbb{N}\rightarrow \mathbb{R}$. An ABC rule has a \emph{proportionality degree} of $f$ if for each instance $(A,k)$, each winning committee $W$, and each $\ell$-cohesive group $V$, the average number of winners that voters from $V$ approve is at least $f(\ell)$, i.e.,
    \begin{equation*}
        \frac{1}{|V|}\sum_{i\in V}|A_i\cap W|\ge f(\ell).
    \end{equation*}
\end{definition}

Since $k=1$ in Example \ref{eg: tiny}, it is not hard to see that $W_1$, $W_2$, $W_3$, $W_4$ achieve the proportionality degree of $1$, $3/4$, $1/2$, $1/4$, respectively. 

Within each cohesive group, the proportionality degree measures the \emph{average satisfaction}, which may not align with the \emph{number of represented voters}.
Specifically, for a cohesive group where not all voters are represented, optimizing the proportionality degree may further increase the satisfaction of voters who are already represented while leaving some other voters completely unrepresented.
The following example demonstrates this.

\begin{example}
    Let $k=3$, $N=\{1,2,\ldots,9\}$, and $C=\{c_1,c_2,\ldots,c_6\}$, where
    \begin{itemize}
        \item $c_1$ is approved by voters $1,2,3,4,5$;
        \item $c_2$ is approved by voters $4,5,6,7,8$;
        \item $c_3$ is approved by voters $7,8,9,1,2$;
        \item each of the candidates $c_4,c_5,c_6$ is approved by voters $1,2,4,5,7,8$.
    \end{itemize}
    In this example, the winner committee that optimizes the proportionality degree is $\{c_4,c_5,c_6\}$, where we have $f(1)=2$ and $f(2)=3$, and the value of $f$ is maximized at both $1$ and $2$ (where $f$ is the function defined in Definition~\ref{def:prop-degree} but with a fixed instance).
    In this case, voters $3,6,9$ are completely unsatisfied.

    On the other hand, under the winner committee $\{c_1,c_2,c_3\}$, every voter is represented in both JR and EJR senses.
    To see this, all voters approve at least one winner, so all voters in each $1$-cohesive group are represented; there is only one $2$-cohesive group $\{1,2,4,5,7,8\}$, and each voter in this cohesive group approves two winners.
    However, this winner committee is sub-optimal in the proportionality degree measurement:
    $f(1)=5/3$ and $f(2)=2$.
\end{example}

In the example above, the winner committee $\{c_4,c_5,c_6\}$ is clearly ``unfair'' to voters $3,6,9$, and the winner committee $\{c_1,c_2,c_3\}$ is much more appealing, at least in some applications.
Indeed, many social choice scenarios require the committee to represent as many voters as possible (e.g., Monroe's rule and the social coverage objective) rather than to increase the average satisfaction. Hence, while maximizing average satisfactions of cohesive groups measured by the proportionality degree is a natural goal in some applications, in some other scenarios, increasing the number of represented voters is a more appealing goal.
Motivated by these, we seek to define a new quantitative notion that describes the \emph{number of represented voters} (instead of the average satisfaction) in each cohesive group.

There is also a natural follow-up question: how can we find the committee that satisfies the optimal (maximum) justified representation degree, given an instance $(N, C, \mathbf{A}, k)$? In this paper, we propose the degree of (E)JR and study its optimization problem.


\subsection{Our New Notions}
In this paper, we study the degree of (E)JR and its optimization problem.

\paragraph{Proposal of new notion---(E)JR degree.} Intuitively, given a ballot instance, we say a winning committee provides \emph{an (E)JR degree of $c$} if at least $c$ voters in every cohesive group are represented. 
\begin{definition}[JR Degree]
    Given a ballot profile $\mathbf{A}=(A_1,\dots, A_n)$ over a candidate set $C$ and a committee size $k$, we say that a set of candidates $W$ of size $|W| = k$ \emph{achieves JR degree $c$} for $(\mathbf{A}, k)$ if, for every cohesive group, there are at least $c$ voters each of which approves at least $1$ candidate in $W$.
\end{definition}

\begin{definition}[EJR degree]
        Given a ballot profile $\mathbf{A}=(A_1,\dots, A_n)$ over a candidate set $C$ and a committee size $k$, we say that a set of candidates $W$ of size $|W| = k$ \emph{achieves EJR degree $c$} for $(\mathbf{A}, k)$ if, for every $\ell$-cohesive group with every $\ell\in[k]$, there are at least $c$ voters each of which approves at least $\ell$ candidates in $W$.
\end{definition}

In this paper, we only consider ballot instances where at least one cohesive group exists, to avoid the uninteresting degenerated case (with no cohesive group) that invalidates the above two definitions.

From the above definitions, a winner committee satisfying JR has a JR degree of at least $1$, and a winner committee satisfying EJR has an EJR degree of at least $1$.
In Example \ref{eg: tiny}, $W_1$, $W_2$, $W_3$, $W_4$ achieve the (E)JR degree of 4,3,2,1, respectively. 

\paragraph{Relationship between JR degree and EJR degree.}
In general, given any ballot instance with $n$ voters and a winner committee of size $k$, one can see that both the JR degree and the EJR degree are at most $\lceil\frac{n}{k}\rceil$: if a cohesive group exists, there is always a cohesive group with size exactly $\lceil\frac{n}{k}\rceil$.
On the other hand, it is widely known that an EJR committee always exists (and so does a JR committee), and many algorithms are known to find an EJR committee \cite{DBLP:journals/scw/AzizBCEFW17,DBLP:conf/aaai/ElkindFIMSS22} some of which run in polynomial time \cite{aziz2018complexity,DBLP:conf/nips/PierczynskiSP21,DBLP:conf/sigecom/Brill023}.
Therefore, for any ballot instance, the maximum degrees for both JR and EJR are at least $1$ and at most $\lceil\frac{n}{k}\rceil$.

Since EJR implies JR, it is easy to see from our definitions that, if a winner committee provides an EJR degree of $c$, its JR degree is at least $c$. 
However, a winner committee may have a higher JR degree than its EJR degree. In addition, the winner committee that maximizes the EJR degree may not be the same as the winner committee that maximizes the JR degree.
The following proposition, whose proof is deferred to Appendix~\ref{append:difference}, shows that the difference between the JR degree and the EJR degree can be significant.

\begin{proposition}\label{prop:difference}
    For any $\gamma>0$, there is a ballot instance with maximum JR degree $c^\ast_{\text{JR}}$ and maximum EJR degree $c^\ast_{\text{EJR}}$ such that
    \begin{enumerate}
        \item $c^\ast_{\text{JR}}/c^\ast_{\text{EJR}}>\gamma$,
        \item all winner committees with EJR degree $c^\ast_{\text{EJR}}$ have JR degrees at most $c^\ast_{\text{JR}}/\gamma$, and
        \item all winner committees with JR degree $c^\ast_{\text{JR}}$ have EJR degrees at most $c^\ast_{\text{EJR}}/\gamma$.
    \end{enumerate}
\end{proposition}

While a committee satisfying EJR is generally considered better than one that only satisfies JR, the JR degree and EJR degree are not directly comparable. Intuitively, the EJR degree faces more challenge of satisfying voters in larger cohesive groups, i.e., $\ell$ approved winners needed for representing a voter in a $\ell$-cohesive group, whereas in the JR degree setting, a voter is represented if there is an approved winner. Depending on the scenario, different degree measurements may be more appropriate. Thus, both metrics are valuable and merit further study.

\paragraph{On optimizing (E)JR degree.}
The definition of (E)JR degree naturally motivates the following optimization problem, which we define as \MDJR/\MDEJR.

\begin{definition}[Maximum (E)JR degree, \MDJR (\MDEJR)]
    Given an instance $(A,k)$, \MDJR (\MDEJR) outputs one committee that achieves the maximum (E)JR degree.
\end{definition}


\paragraph{``Number'' versus ``fraction''}
We have defined the JR and EJR degrees based on the \emph{number} of represented voters.
Another natural way is to define both notions based on the \emph{fraction} of represented voters.

When dealing with the JR degree, both definitions are equivalent.
To see this, in terms of both numbers and fractions, the least satisfying cohesive group always contains exactly $\lceil n/k\rceil$ voters, as removing a represented voter from a cohesive group of more than $\lceil n/k\rceil$ voters would make both the number and the fraction of represented voters decrease.
Given that we are concerning the least satisfying cohesive group (which always has the same size), the ``number version'' and the ``fraction version'' of the JR degree are equivalent.

When dealing with the EJR degree, the two definitions are different.
The ``number version'' fits better with the spirit of EJR.
In the definition of EJR (see Definition~\ref{def:EJR}), the requirement is that \emph{one} voter needs to be represented (approve at least $\ell$ winners) in every $\ell$-cohesive group, instead of being that $\ell$ voters need to be represented in every $\ell$-cohesive group.
On the other hand, for some $\ell>1$, given an $\ell$-cohesive group with minimum size $\ell\cdot\frac nk$ and a $1$-cohesive group with minimum size $\frac nk$, if one voter is represented in the $1$-cohesive group, then $\ell$ voters in the $\ell$-cohesive group need to be represented in order to make the $\ell$-cohesive as ``happy'' as the $1$-cohesive group in the case the EJR degree is defined in the ``fraction version''.

\subsection{Our Technical Contributions}
In this paper, we focus on the computational complexity and the approximability of the optimization problems \MDJR and \MDEJR.
Our results are listed below.

\begin{itemize}
    \item We first show that the algorithm for finding a JR committee proposed by~\citet{DBLP:journals/scw/AzizBCEFW17} also provides a $\frac1k$-approximation to \MDJR, and the algorithm for finding an EJR committee proposed by~\citet{aziz2018complexity} also provides a $\frac1{k+1}$-approximation to \MDEJR.
    On the other hand, since the maximum (E)JR degree for both \MDJR and \MDEJR is $\lceil\frac nk\rceil$, the approximation guarantees for the two algorithms above can also be written as $1/\lceil\frac nk\rceil\approx k/n$.
    To complement these positive results, we show almost tight inapproximability results. We show that it is NP-hard to approximate \MDJR (\MDEJR) within a factor of $(k/n)^{1-\epsilon}$ for any $\epsilon>0$. We also show that it is NP-hard to approximate \MDJR (\MDEJR) within a factor of $(1/k)^{1-\epsilon}$ for any $\epsilon>0$.
    \item We study the fixed-parameter-tractability of this problem. We show that finding a committee with the maximum achievable (E)JR degree is W[2]-hard if $k$, the size of the winning committee, is specified as the parameter. 
    \item When the maximum achievable (E)JR degree of an instance is additionally given as a parameter, we show that the problem is fixed-parameter-tractable.
\end{itemize}

Surprisingly, although Proposition~\ref{prop:difference} demonstrates that JR degree and EJR degree have different natures, we obtain the same set of results for \MDJR and \MDEJR.

\subsection{Further Related Work}
\label{sec:relatedwork}
In this subsection, we discuss the related work on justified representation (JR). In addition to the JR axioms previously mentioned, several other JR-related axioms have been proposed and studied. \citet{DBLP:conf/aaai/FernandezELGABS17} introduced Proportional JR (PJR), which requires that every $\ell$-cohesive group have some $\ell$ winners represented in the union of their approval sets. PJR is weaker than Extended JR (EJR) but stronger than JR. The authors also proposed Perfect Representation (PR), aiming to represent all voters by some winners, with each winner representing $\frac{n}{k}$ voters. While a PR committee may not always exist, if one does, it can be verified that such a committee guarantees the maximum JR degree of $\frac{n}{k}$.

\citet{DBLP:conf/nips/PierczynskiSP21} introduced Fully JR (FJR), which weakens the cohesiveness requirement. It considers groups of $\ell\frac{n}{k}$ voters who share at least $\beta \leq \ell$ candidates in common. A committee satisfies FJR if every $(\ell, \beta)$-weak-cohesive group, where $\ell \in [k]$ and $\beta \in [\ell]$, has at least one member who approves $\beta$ winners. Notably, FJR implies EJR. \citet{DBLP:conf/sigecom/Brill023} proposed PJR+ and EJR+, which focus on ensuring PJR and EJR for groups of $\ell \frac{n}{k}$ voters, where at least one non-elected candidate is approved by all group members, as opposed to considering only $\ell$-cohesive groups. Consequently, EJR+ implies both EJR and PJR+, while PJR+ implies PJR.

\citet{DBLP:conf/aaai/BrillIM022} studied Individual Representation (IR), which requires that every voter in an $\ell$-cohesive group is represented. An IR committee would achieve the maximum (E)JR degree, although such committees may not always exist. The potential non-existence of IR and PR committees, compared to the guaranteed existence of (E)JR committees, motivates us to explore a quantitative measure bridging (E)JR, PR, and IR.

Moreover, JR has been investigated in other domains, such as fair division \cite{DBLP:conf/aaai/BeiLS22,DBLP:conf/aaai/Lu00BS23}, participatory budgeting \cite{DBLP:conf/atal/0001LT18,DBLP:conf/nips/PierczynskiSP21}, and facility location games \cite{DBLP:conf/atal/ElkindLZ22}. Other properties that assess a committee's proportionality, such as laminar proportionality and priceability, have also been considered \cite{DBLP:conf/sigecom/PetersS20}.

In our work, we study voting rules that maximize the (E)JR degree, particularly \MDJR and \MDEJR. In multi-winner voting, there are a variety of voting rules that maximize certain scores, collectively known as Thiele methods \cite{Theiele1895}. Thiele methods focus on maximizing the sum of voters' individual satisfaction, where a voter's satisfaction is determined by the number of approved candidates in the winning committee. While our \MDEJR maximization draws inspiration from this concept, \MDJR maximization fundamentally differs, as it focuses on maximizing the JR degree within each cohesive group rather than global satisfaction.

Monroe's rule \cite{monroe1995fully} shares a similar objective with our \MDJR approach, as it seeks to maximize the number of voters represented by at least one candidate in the winning committee. This mirrors our goal of maximizing the number of voters represented by at least one winner in every cohesive group. Other notable voting rules, such as Phragmén’s rules \cite{janson2016phragmen}, have also been studied in the context of multi-winner voting \cite{DBLP:series/sbis/LacknerS23}.

Finally, several well-known properties in multi-winner approval voting are relevant to our study, including anonymity and neutrality \cite{may1952set,moulin1991axioms,arrow2012social}, Pareto efficiency \cite{DBLP:journals/ai/LacknerS20}, monotonicity \cite{DBLP:conf/atal/FernandezF19}, consistency \cite{DBLP:journals/jet/LacknerS21}, and strategyproofness \cite{DBLP:conf/atal/Peters18}.

\section{Inapproximability of {\tt\textbf{MDJR}} and {\tt\textbf{MDEJR}}} \label{sec: np}
In this section, we study the approximability of \MDJR and \MDEJR.
We provide almost tight approximability of both problems, in terms of both $k$ and $k/n$.

As we mentioned before, the maximum degree for both \MDJR and \MDEJR is $\lceil\frac nk\rceil$.
On the other hand, any winner committee satisfying JR (EJR, resp.) gives a JR degree (EJR degree, resp.) of at least $1$.
Therefore, the algorithm for finding a JR (EJR, resp.) committee provides a $1/\lceil\frac nk\rceil\approx\frac kn$ approximation to \MDJR (\MDEJR, resp.).
In Sect.~\ref{sec:k-approx}, we will show that the algorithm for finding a JR committee proposed by~\citet{DBLP:journals/scw/AzizBCEFW17} also provides a $\frac1k$-approximation to \MDJR, and the algorithm for finding an EJR committee proposed by~\citet{aziz2018complexity} also provides a $\frac1{k+1}$-approximation to \MDEJR.

In Sect.~\ref{sec:inapproximability}, we show that the approximation ratios of approximately $k/n$ and $1/k$ mentioned above are almost tight.

\subsection{Approximation Algorithms}
\label{sec:k-approx}
The algorithm GreedyAV (Algorithm~\ref{alg: greedyAV}) proposed by~\citet{DBLP:journals/scw/AzizBCEFW17} always outputs a JR committee. We will show in Theorem~\ref{thm:MDJR-approximation} that it provides a $\frac1k$-approximation to \MDJR.

\begin{algorithm}[!htb]
\caption{Greedy Approval Voting (GreedyAV)}
\label{alg: greedyAV}
\KwIn{An instance $\mathcal{I} = (N, C, \mathbf{A}, k))$}
\KwOut{A winning committee $W$ of size $k$}
$W\gets \emptyset$

\For{$j\in [k]$}{
Let $c$ be the candidate approved by the maximum number of voters in $N$

$W\gets W\cup c$

$N\gets N\setminus V$ where $V$ is the set of voters who approve $c$
}
\Return $W$
\end{algorithm}

We prove the following proposition first.

\begin{proposition}\label{prop:MDJR-approximation}
    GreedyAV outputs a committee achieving JR degree at least $\frac{n}{k^2}$.
\end{proposition}
\begin{proof}
    We prove it by contradiction. Suppose the committee output by GreedyAV provides the JR degree less than $\frac{n}{k^2}$. Hence, there will be at least $(\frac{n}{k}-\frac{n}{k^2}+1)$ voters approving no candidate in $W$ that approve a common candidate $c\notin W$. From the definition of GreedyAV, the coverage of voters is at least $(\frac{n}{k}-\frac{n}{k^2}+1)$ in each iteration, and at least $\frac{n}{k}$ in the first iteration (given that there is at least one cohesive group). Therefore, the total number of voters is at least
    \begin{equation*}\begin{aligned}
        \frac{n}{k}+(k-1)\left(\frac{n}{k}-\frac{n}{k^2}+1\right)+\left(\frac{n}{k}-\frac{n}{k^2}+1\right) =n+\frac{n}{k}-k\left(\frac{n}{k^2}-1\right)>n,
    \end{aligned}
    \end{equation*}
    which leads to a contradiction in the number of voters.
\end{proof}

Now, we are ready to show the approximation guarantee of GreedyAV.

\begin{theorem}\label{thm:MDJR-approximation}
    GreedyAV runs in polynomial time and provides a $\frac1k$-approximation to \MDJR.
\end{theorem}
\begin{proof}
    It is clear that the algorithm runs in polynomial time.
    To show the approximation guarantee, we show that the minimum JR degree of $\frac{n}{k^2}$ proved in Proposition~\ref{prop:MDJR-approximation}.
    To see this, notice that we have remarked that the maximum possible JR degree is $\lceil\frac nk\rceil$.
    On the other hand, since the JR degree is an integer, the minimum JR degree that Greedy AV can guarantee is actually $\lceil\frac{n}{k^2}\rceil$.
    It suffices show that
    $$\left\lceil\frac{n}{k^2}\right\rceil\geq\frac1k\cdot\left\lceil\frac nk\right\rceil.$$
    The inequality clearly holds if $n$ is a multiple of $k$.
    Suppose otherwise. Let $n=ak+b$ for non-negative integers $a,b$ with $1\leq b\leq k-1$.
    Let $a=sk+t$ for non-negative integers $s,t$ with $0\leq t\leq k-1$.
    For the left-hand side of the inequality, we have
    $$\left\lceil\frac{n}{k^2}\right\rceil=\left\lceil\frac{(sk+t)k+b}{k^2}\right\rceil\geq\left\lceil s+\frac{b}{k^2}\right\rceil\geq s+1,$$
    where the first inequality above is due to $t\geq0$ and the second inequality above is due to $b\geq1$.
    For the right-hand side of the inequality, we have
    $$\frac1k\cdot\left\lceil\frac nk\right\rceil=\frac1k\cdot(a+1)=\frac{sk+t+1}{k}\leq s+1, $$
    where the last inequality is due to that $t\leq k-1$.

    Putting together, we have $\left\lceil\frac{n}{k^2}\right\rceil\geq\frac1k\cdot\left\lceil\frac nk\right\rceil$, which implies the theorem.
\end{proof}

Similar to \MDJR, we first consider an EJR voting rule, \emph{proportional approval voting} (PAV) \cite{DBLP:journals/scw/AzizBCEFW17}, which outputs the committee that maximizes the PAV-score, where the PAV-score of a committee $W \subseteq C$ is defined as
\begin{equation*}
    s_{\text{PAV}}(W)=\sum_{i=1}^n\sum_{j=1}^{|A_i\cap W|} \frac{1}{j}.
\end{equation*}
In the PAV-score, each voter's ``utility'' is defined by the harmonic progression $H[t]$ for $t$ approved winning candidates and the PAV-score can then be understood as the \emph{social welfare}.
One may wonder whether PAV can provide the maximum EJR degree. We find that PAV fails to achieve the maximum EJR degree in some instances. A counterexample can be found in Appendix~\ref{exxample : PAV fail mdejr}.

In addition, PAV cannot be computed in polynomial time. \citet{aziz2018complexity} showed that a local search alternative algorithm for PAV can both satisfy EJR and be computed in polynomial time. 
The algorithm is described in Algorithm~\ref{alg:LS-PAV}. Starting from an arbitrary winning committee, the algorithm considers all possible single-candidate-replacements that increase the PAV score by at least $\lambda$ (where $\lambda$ is a parameter of the algorithm).
For each pair of candidate $c^+$ and $c^-$ with $c^+\notin W$ and $c^-\in W$, if we swap $c^+$ and $c^-$, i.e. to remove $c^-$ from the committee and select $c^+$ instead, the score is increased by 
$
    \Delta(W,c^+,c^-)=s_{\text{PAV}}(W \setminus \{c^-\} \cup \{c^+\})-s_{\text{PAV}}(W).
$

\begin{algorithm}[!htb]
\caption{$\lambda$-LS-PAV}
\label{alg:LS-PAV}
\KwIn{An instance $\mathcal{I} = (N, C, \mathbf{A}, k))$}
\KwOut{A winning committee $W$ of size $k$}
$W\gets k$ arbitrary candidates from $C$

\While{there exist $c^+\notin W$ and $c^-\in W$ such that $\Delta(W,c^+,c^-)\ge \lambda$}{
    $W \gets W \setminus \{c^-\} \cup \{c^+\}$
}
\Return $W$
\end{algorithm}

Next, we will show that $\lambda$-LS-PAV runs in polynomial time and provides a $\frac1{k+1}$-approximation to \MDEJR for a suitable choice of $\lambda$.

We show the following proposition first.
\begin{proposition}\label{prop:MDEJR-approximation}
    $\lambda$-LS-PAV outputs a committee achieving EJR degree at least $c^\ast=\frac{n}{k(k+1)}-\lambda\frac{k}{k+1}$ for any $\lambda\in[0,\frac n{k^2})$.
\end{proposition}
\begin{proof}
    We prove it by contradiction. Suppose the committee, $W$, output by $\lambda$-LS-PAV provides the EJR degree strictly smaller than $c^\ast$. There exists a $\ell$-cohesive group $V\subseteq N$, such that less than $c^\ast$ voters in $V$ approves at least $\ell$ members of $W$, i.e., $|\{i\in V:|A_i\cap W|\ge \ell\}|<c^\ast$. Since $V$ is $\ell$-cohesive, there exist $\ell$ candidates approved by all voters in $V$. At least one such candidate, $c^+\in \bigcap_{i\in V} A_i$, is not selected, as otherwise all voters in $V$ approve at least $\ell$ members of $W$. 
    
    We will show that there exists a candidate $c^-\in W$ such that the increment of the score by swapping $c^+$ and $c^-$ is at least $\lambda$, so $W$ cannot be an output of $\lambda$-LS-PAV, leading to a contradiction. To see this, we try to swap $c^+$ and any candidate $c^-\in W$. Since $c^+\in A_i$ for all $i\in V$, we have
    \begin{equation*}
    \begin{aligned}
        \Delta(W,c^+,c^-) =&\sum_{\substack{i:c^+\in A_i\\c^- \notin A_i}}\frac{1}{|A_i\cap W|+1}-\sum_{\substack{i:c^+\notin A_i\\c^- \in A_i}}\frac{1}{|A_i\cap W|}\\
        \ge& \sum_{i\in V:c^- \notin A_i}\frac{1}{|A_i\cap W|+1}-\sum_{i\in N\setminus V:c^- \in A_i}\frac{1}{|A_i\cap W|},
    \end{aligned}
    \end{equation*}
    and, by summing up $\Delta(W,c^+,c^-)$ for $c^-\in W$,
    \begin{equation*}\begin{aligned}
         \sum_{c^-\in W} \Delta(W,c^+,c^-) \ge &\sum_{c^-\in W}\sum_{\substack{i\in V: \\c^- \notin A_i}}\frac{1}{|A_i\cap W|+1}- \sum_{c^-\in W}\sum_{\substack{i\in N\setminus V:\\c^- \in A_i}}\frac{1}{|A_i\cap W|}\\
        = &\sum_{i\in V} \sum_{c^-:\substack{c^-\in W\\c^-\notin A_i}} \frac{1}{|A_i\cap W|+1} -\sum_{i\in N\setminus V} \sum_{c^-:\substack{c^-\in W\\c^-\in A_i}} \frac{1}{|A_i\cap W|}\\
        = &\left(\sum_{i\in V} \frac{k-|A_i\cap W|}{|A_i\cap W|+1}\right) - \left(n -|V|\right)\\
        = &\left(\sum_{i\in V} \frac{k+1}{|A_i\cap W|+1}\right) - n\\
        \ge &\left(\sum_{i\in V:|A_i\cap W|<\ell} \frac{k+1}{|A_i\cap W|+1}\right) - n\\
        \ge &\left(\sum_{i\in V:|A_i\cap W|<\ell} \frac{k+1}{\ell} \right)- n.
    \end{aligned}\end{equation*}
    Since less than $c^\ast$ voters in $V$ approve at least $\ell$ members of $W$, the number of voters in $V$ approve less than $\ell$ members of $W$ is at least $|V|-c^\ast+1$. Thus,
    \begin{equation*}\begin{aligned}
        \sum_{c^-\in W} \Delta(W,c^+,c^-) &\ge \left(\sum_{i\in V:|A_i\cap W|<\ell} \frac{k+1}{\ell}\right) - n \\
        &\ge (k+1)\frac{|V|-c^\ast+1}{\ell} - n\\
        &>(k+1)\left(\frac{|V|}\ell-c^\ast\right)-n\\
        &\geq(k+1)\left(\frac{n}k-c^\ast\right)-n\\
        &=\frac{n}{k}-(k+1)c^\ast\\
        &=\lambda k.
    \end{aligned}\end{equation*}
    From the pigeonhole principle, it follows that there exists a candidate $c^-\in W$ such that $\Delta(W,c^+,c^-)\ge\lambda$.
\end{proof}

Now, we are ready to conclude the approximation guarantee of $\lambda$-LS-PAV.

\begin{theorem}\label{thm:MDEJR-approximation}
    For $\lambda=\frac1{2k^2}$, $\lambda$-LS-PAV runs in polynomial time and provides a $\frac1{k+1}$-approximation to \MDEJR.
\end{theorem}
\begin{proof}
    Firstly, the algorithm runs in polynomial time: the maximum possible PAV-score is $n\cdot\sum_{j=1}^k\frac1j=O(n\log k)$.
    On the other hand, the PAV-score is increased by at least $\frac1{2k^2}$.
    The algorithm terminates in at most $O(nk^2\log k)$ iterations, which runs in polynomial time.

    To prove the approximation guarantee, Proposition~\ref{prop:MDEJR-approximation} shows that the EJR degree of the output committee is at least $\frac{n}{k(k+1)}-\frac{1}{2k(k+1)}$.
    Since the degree is an integer, this is at least $\lceil\frac{n}{k(k+1)}\rceil$: if $\frac{n}{k(k+1)}$ is an integer, we have $\lceil\frac{n}{k(k+1)}-\frac{1}{2k(k+1)}\rceil=\frac{n}{k(k+1)}$ as $\frac{1}{2k(k+1)}<1$; otherwise, since $\frac{n}{k(k+1)}-\lfloor\frac{n}{k(k+1)}\rfloor\geq\frac1{k(k+1)}>\frac1{2k(k+1)}$, the numbers $\frac{n}{k(k+1)}$ and $\frac{n}{k(k+1)}-\frac{1}{2k(k+1)}$ have the same integral part.

    Following a similar analysis as in the proof of Theorem~\ref{thm:MDJR-approximation}, we can prove that
    $$\left\lceil\frac n{k(k+1)}\right\rceil\geq\frac1{k+1}\cdot\left\lceil\frac nk\right\rceil,$$
    which implies the theorem.

    Specifically, the inequality clearly holds if $k$ divides $n$.
    Otherwise, let $n=ak+b$ for integers $a\geq 0$ and $1\leq b\leq k-1$.
    Let $a=s(k+1)+t$ for integers $s\geq0$ and $0\leq t\leq k$.
    We have
    $$\left\lceil\frac n{k(k+1)}\right\rceil=\left\lceil\frac {sk(k+1)+tk+b}{k(k+1)}\right\rceil\geq s+1\geq\frac{s(k+1)+t+1}{k+1}=\frac1{k+1}\cdot (a+1)=\frac1{k+1}\cdot\left\lceil\frac nk\right\rceil.\qedhere$$
\end{proof}

\subsection{Matching Inapproximability Results}
\label{sec:inapproximability}
To complement the positive results mentioned in the previous section, 
we present the following inapproximability results.

\begin{theorem}\label{thm:NPhard_JR}
    It is NP-hard to approximate \MDJR within a factor of $\left(\frac{k}{n}\right)^{1-\epsilon}$ for any $\epsilon>0$.
\end{theorem}

\begin{theorem}\label{thm:NPhard_EJR}
    It is NP-hard to approximate \MDEJR within a factor of $\left(\frac{k}{n}\right)^{1-\epsilon}$ for any $\epsilon>0$.
\end{theorem}

\begin{theorem}\label{thm:NPhard_JR-k}
    It is NP-hard to approximate \MDJR within a factor of $\left(\frac{1}{k}\right)^{1-\epsilon}$ for any $\epsilon>0$.
\end{theorem}

\begin{theorem}\label{thm:NPhard_EJR-k}
    It is NP-hard to approximate \MDEJR within a factor of $\left(\frac{1}{k}\right)^{1-\epsilon}$ for any $\epsilon>0$.
\end{theorem}

We will simultaneously prove these four theorems by constructing a hard ballot instance that is used for all of them.
We will make sure the instance we constructed has no $\ell$-cohesive group with $\ell>1$.
Notice that the JR degree and the EJR degree for any committee are always the same for instances with only $1$-cohesive groups.
In addition, we have $n\approx k^2$ in our construction, so the factor $k/n$ is approximately $1/k$.

Before we prove the theorems, we first introduce a NP-hard problem: \NEWSAT problem. One can find variations of SAT problem and a similar argument in \citet{tovey1984simplified}.
\begin{definition}[\NEWSAT]
    Given a CNF formula $\phi$ that, for any pair of variables $x$ and $y$, at most one clause contains both $x$ (or $\neg x$) and $y$ (or $\neg y$), decide if there is a value assignment to the variables to make $\phi$ true.
\end{definition}

To see its NP-hardness, it can be reduced from the SAT problem. Start with any SAT instance. Without loss of generality, suppose each variable, $x$ or $\neg x$, appears in each clause at most once. For each variable $x$ such that $x$ or $\neg x$ is contained in more than two clauses, we perform the following procedure. Suppose $x$ appears in $k$ clauses. Create $k$ new variables $x_1,\dots,x_k$ and replace the $i$th occurrence of $x$ with $x_i$ (and $\neg x$ is replaced by $\neg x_i$, respectively) for each $i=1,\dots,k$. Add the clause $(x_i\vee \neg x_{i+1})$ for $i=1,\dots,k-1$ and the clause $(x_k\vee \neg x_1)$. Note that, in the new instance, variable $x_i$ and $y_j$ appear in a clause only when the $i$th occurrence of $x$ and the $j$th occurrence of $y$ in $\phi$ are in the same clause, so the new instance satisfies the requirement of the \NEWSAT problem.

In the new instance, the clause $(x_i\vee \neg x_{i+1})$ implies that if $x_i$ is false, $x_{i+1}$ must be false as well. The cyclic structure of the clauses therefore forces the $x_i$ to be either all true or all false, so the new instance is satisfiable if the original one is. Moreover, the transformation requires polynomial time.

Now we are ready to prove our theorems. 

\begin{figure}
    \centering
    \begin{tikzpicture}[scale=0.4]
    \node[left] at (-2.5,0) {Voters};
    \node[left] at (-2.5,-8) {Candidates};

    \foreach \x in {1,2,3,4,5,6,7} {
        \draw[thick] (5*\x-5,0) ellipse (2 and 1);
        \filldraw[black] (5*\x-6,0) circle (0.15);

        \filldraw[black] (5*\x-4,0) circle (0.15);
    }

    \foreach \x in {1,2,3,4,5,7} {

        \filldraw[black] (5*\x-5,0) circle (0.15);

    }
    
    \foreach \x in {1,2,3} {
        \node[above] at (5*\x-5,1.2) {$T_\x$};
        \draw[thick] (5*\x-5,-1) -- (5*\x-5,-7);
        \draw[thick] (5*\x-5,-8) ellipse (2 and 1);
        \filldraw[black] (5*\x-5.8,-8) circle (0.15);
        \filldraw[black] (5*\x-4.2,-8) circle (0.15);
    }

    \foreach \x in {1,2,3,4,5,6} {
        \filldraw[black] (3*\x+12,-8) circle (0.15);
    }

    \node[above] at (15,1.2) {$S_1$};
    \node[above] at (20,1.2) {$S_2$};
    \node[above] at (25,1.2) {$D$};    
    \node[above] at (30,1.2) {$D^+$};

    \node[above] at (-0.8,-10.5) {$c_1$};
    \node[above] at (-0.8,-11.2) {\small $x_1$};
    \node[above] at (0.8,-10.5) {$c_2$};
    \node[above] at (0.8,-11.2) {\small $\neg x_1$};
    \node[above] at (4.2,-10.5) {$c_3$};
    \node[above] at (4.2,-11.2) {\small $x_2$};
    \node[above] at (5.8,-10.5) {$c_4$};
    \node[above] at (5.8,-11.2) {\small $\neg x_2$};
    \node[above] at (9.2,-10.5) {$c_5$};
    \node[above] at (9.2,-11.2) {\small $x_3$};
    \node[above] at (10.8,-10.5) {$c_6$};
    \node[above] at (10.8,-11.2) {\small $\neg x_3$};
    \node[above] at (15,-10.5) {$t_1$};
    \node[above] at (18,-10.5) {$t_2$};
    \node[above] at (21,-10.5) {$t_3$};
    \node[above] at (24,-10.5) {$s_1$};
    \node[above] at (27,-10.5) {$s_2$};
    \node[above] at (30,-10.5) {$d$};
    \node[above] at (15, -13.2) {$\phi=(x_1 \vee \neg x_3)\wedge(\neg x_1 \vee x_2)$, $\bar{m}=3$,  $\bar{n}=2$, $k=\bar{m}+1=4$.};
    \draw[thick] (0,-1) -- (15,-8);
    \draw[thick] (5,-1) -- (18,-8);
    \draw[thick] (10,-1) -- (21,-8);
    
    \draw[thick] (15,-1) -- (-0.8,-8);
    \draw[thick] (15,-1) -- (10.8,-8);
    \draw[thick] (15,-1) -- (24,-8);    

    \draw[thick] (20,-1) -- (0.8,-8);
    \draw[thick] (20,-1) -- (4.2,-8);
    \draw[thick] (20,-1) -- (27,-8);    

    \draw[thick] (25,-1) -- (15,-8);
    \draw[thick] (25,-1) -- (18,-8);
    \draw[thick] (25,-1) -- (21,-8); 
    \draw[thick] (25,-1) -- (24,-8);
    \draw[thick] (25,-1) -- (27,-8); 
    \draw[thick] (25,-1) -- (30,-8);  

    \draw[thick] (30,-1) -- (30,-8); 
\end{tikzpicture}
\caption{An example of the construction with $\phi=(x_1\vee\neg x_3)\wedge(\neg x_1\vee x_2)$. Edges in the graph represent approvals. An edge connecting between a group of voters and a group of candidates indicates that every voter in the voter group approves every candidate in the candidate group.}
\label{fig:reduction}
\end{figure}
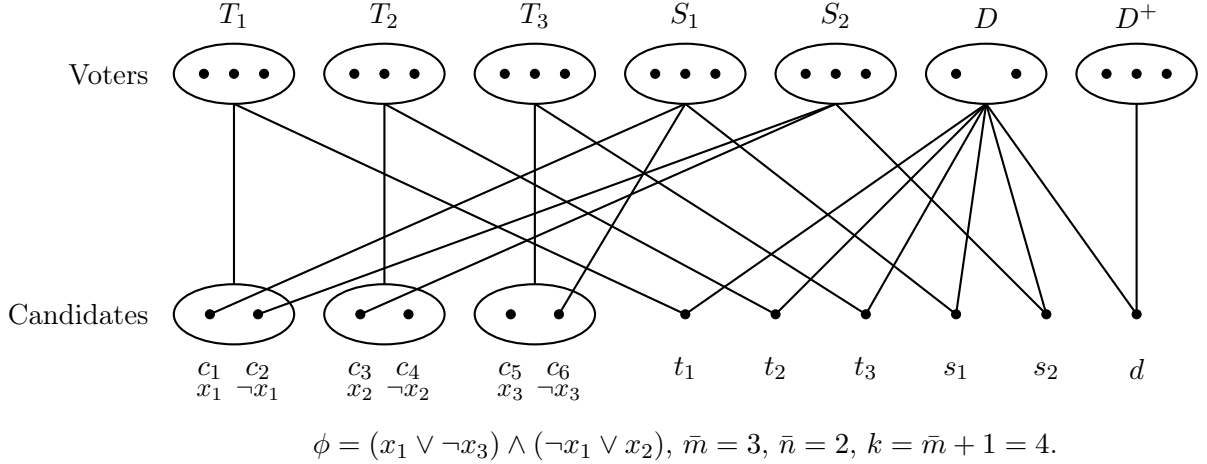

\begin{proof}[Proof of Theorem~\ref{thm:NPhard_JR} and \ref{thm:NPhard_EJR}]\label{proof_of_thm1/2}
    We reduce from \NEWSAT problem. Given any \NEWSAT instance $\phi$, suppose there are $\bar{n}$ clauses and $\bar{m}$ variables (say $x_1,\dots,x_{\bar{m}}$). We consider an ABC voting instance with $3\bar{m}+\bar{n}+1$ candidates $\{c_1,\dots,c_{2\bar{m}}$, $s_1,\dots,s_{\bar{n}}$, $t_1,\dots, t_{\bar{m}}, d\}$ and $\bar{n}\bar{m}+\bar{m}^2+\bar{n}+\bar{m}$ voters. We want to select a committee of size $\bar{m}+1$. 
    Hence, we care about the cohesive group of size $\frac{\bar{n}\bar{m}+\bar{m}^2+\bar{n}+\bar{m}}{\bar{m}+1}=\bar{n}+\bar{m}$. 
    First, for each variable $x_j$ and its negation $\neg x_j$, we create two corresponding candidates $c_{2j-1}$ and $c_{2j}$ and a group $T_j$ of $\bar{m}$ voters who approve $c_{2j-1}$ and $c_{2j}$. For the $i$th clause, we create a group $S_i$ of $\bar{m}$ voters. All voters in group $S_i$ approve $c_{2j-1}$ if $x_j$ occurs in the $i$th clause and $c_{2j}$ if $\neg x_j$ occurs in the $i$th clause. All voters in group $S_i$ approve $s_i$ and voters in group $T_i$ approve $t_i$ additionally. We create a set $D$ of $\bar{n}$ voters who approve $s_1,\dots,s_{\bar{n}},t_1,\dots,t_{\bar{m}},d$. Hence, for each $i\in [\bar{n}]$, $S_i\cup D$ forms a $1$-cohesive group, and for each $j\in [\bar{m}]$, $T_j\cup D$ forms a $1$-cohesive group. Moreover, we create a set $D^+$ of $\bar{m}$ voters who approve $d$. Hence, $D\cup D^+$ forms a $1$-cohesive group.
    An example of our construction is shown in Fig.~\ref{fig:reduction}.
    
    Notably, there is no $2$-cohesive group. First, candidate $s_i$ for $i\in [\bar{n}]$ or $t_j$ for $j\in [\bar{m}]$ or $d$ has only $\bar{n}+\bar{m}$ voters approving them as constructed above. For candidates $c_1,\dots,c_{2\bar{m}}$, we will show that no $2(\bar{n}+\bar{m})$ voters have two common approved candidates. 
    For any two candidates $t_{2i-1},t_{2i}$ in $\{c_1,\dots,c_{2\bar{m}}\}$ that correspond to a variable and its negation, the set of voters approving both candidates is exactly $T_i$ (notice that we can assume without loss of generality that $x_i$ and $\neg x_i$ do not appear in the same clause, for otherwise, we can safely remove both literals).
    We have $|T_i|=\bar{m}<2(\bar{n}+\bar{m})$.
    For any two candidates in $\{c_1,\dots,c_{2\bar{m}}\}$ that correspond to different variables, by our sparsity assumption of $\phi$, the set of voters approving both candidates is at most some $S_j$ (representing a clause that contains both variables).
    We have $|S_j|=\bar{m}<2(\bar{n}+\bar{m})$.
    Thus, there is no $2$-cohesive group or $\ell$-cohesive group for $\ell>1$. In this case, the EJR degree is equal to the JR degree.
    Thereafter, we will analyze the JR degree and \MDJR, and the analysis applies to \MDEJR as well.

    If there is a value assignment to the variables to make $\phi$ true, then \MDJR will achieve a JR degree of $\bar{n}+\bar{m}$. To see this, for each $j\in [\bar{m}]$, we select $c_{2j-1}$ as a winner if $x_j$ is assigned true, or we select $c_{2j}$. Then we select $d$ as the winner. We can verify that each voter in $\{S_i\}_{i\in [\bar{n}]}\cup\{T_j\}_{j\in [\bar{m}]}$ approves at least one winner. In addition, $d$ is approved by all voters in $D$ and $D^+$. Hence, all voters approve at least one winner, and the JR degree equals the size of a $1$-cohesive group, which is $\bar{n}+\bar{m}$.
    
    If there does not exist a satisfying assignment to $\phi$, then \MDJR will achieve the JR degree of at most $\bar{n}$. To see this, we prove it by contradiction. Assume that \MDJR can achieve the JR degree of larger than $\bar{n}$ by the winner committee $W$ with $|W|=k=\bar{m}+1$. If $d\notin W$, no voter in $D^+$ can be covered since they only approve $d$. Hence, for the cohesive group $D\cup D^+$, at most $\bar{n}$ voters can be covered, which leads to a contradiction. Thus, $d\in W$. Let $W'=W\setminus\{d\}$, and we have $|W'|=\bar{m}$. If voters in $T_j$ are not covered for some $j\in [\bar{m}]$, at most $\bar{n}$ voters can be covered in the cohesive group $D\cup T_j$, leading to a contradiction. Thus, all voters in group $T_j$ must be covered for each $j\in [\bar{m}]$, indicating that exactly one of the $3$ candidates $\{c_{2j-1},c_{2j},t_j\}$ for each $j\in[\bar{m}]$ is selected (at least one candidate in each three-candidates group must be selected, we have $\bar{m}$ groups, and we have $|W'|=\bar{m}$). Next, we will show that voters in at least one group among $S_1,\dots, S_{\bar{n}}$ cannot be covered. 
    Suppose this is not the case.
    For every $t_j\in W'$, we can use either $c_{2j-1}$ or $c_{2j}$ to replace $t_j$ since $t_j$ only covers group $T_j$ and each of $c_{2j-1}$ and $c_{2j}$ covers at least $T_j$. Hence, we can find one candidate in $\{c_{2i-1},c_{2i}\}$ for all $i\in [\bar{m}]$ to cover all groups, implying that there exists a value assignment to the variables ($x_i$ is assigned true if $c_{2i-1}$ is selected, or false otherwise) to make $\phi$ true, leading to a contradiction. Therefore, at least one group among $S_1,\dots, S_{\bar{n}}$ cannot be covered. Without loss of generality, we assume that it is $S_i$. Then, for the cohesive group $S_i\cup D$, at most $\bar{n}$ voters can be covered, leading to a contradiction.

    From the proof of NP-hardness, we have that \MDJR cannot be approximated in polynomial time to within a factor of $\frac{\bar{n}}{\bar{n}+\bar{m}}$, where $\bar{n}$ is the number of the clauses and $\bar{m}$ is the number of the variables. 
    Now, instead of reducing from \NEWSAT problems directly, we add an intermediate reduction. Given a \NEWSAT instance $\phi$ with $\bar{n}$ clauses and $\bar{m}$ different variables, we construct another \NEWSAT instance $\phi'$ with $\bar{n}'=\bar{n}+1$ clauses and $\bar{m}'=\bar{m}+(\bar{n}+\bar{m}+1)^{\lceil\frac{1}{\epsilon}\rceil}$ different variables such that a new clause is added with $(\bar{n}+\bar{m}+1)^{\lceil\frac{1}{\epsilon}\rceil}$ new variables. Obviously, $\phi$ and $\phi'$ have the same satisfiability.

    Now, we use $\phi'$ instead of $\phi$ to construct the ABC voting instance.
    By our previous analysis, \MDJR cannot be approximated in polynomial time to within a factor of $\frac{\bar{n}'}{\bar{n}'+\bar{m}'}$ where $\bar{n}'=\bar{n}+1, \bar{m}'=\bar{m}+(\bar{n}+\bar{m}+1)^{\lceil\frac{1}{\epsilon}\rceil}$. Recall that the number of voters and the committee size are  $n_{\text{voting}}=\bar{n}'\bar{m}'+\bar{m}'^2+\bar{n}'+\bar{m}'$ and $k_{\text{voting}}=\bar{m}'+1$, respectively, which can be reformulated as $\frac{k_{\text{voting}}}{n_{\text{voting}}}=\frac{1}{\bar{n}'+\bar{m}'}$.
    Since
    \begin{equation*}
        (\bar{n}'+\bar{m}')^{\epsilon}=\left(\bar{n}'+\bar{m}+(\bar{n}'+\bar{m})^{\lceil\frac{1}{\epsilon}\rceil}\right)^{\epsilon}>\bar{n}'+\bar{m}>\bar{n}',
    \end{equation*}
    \MDJR cannot be approximated in polynomial time to within a factor of
    \begin{equation*}
        \frac{\bar{n}'}{\bar{n}'+\bar{m}'}\le \frac{(\bar{n}'+\bar{m}')^{\epsilon}}{\bar{n}'+\bar{m}'} = \left(\frac{1}{\bar{n}'+\bar{m}'}\right)^{1-\epsilon} \!\!\!=\! \left(\frac{k_{\text{voting}}}{n_{\text{voting}}}\right)^{1-\epsilon},
    \end{equation*}
    where $\epsilon>0$.

    As we have mentioned, the same analysis holds for \MDEJR, as there is no $\ell$-cohesive group with $\ell\geq 2$ in our construction.
\end{proof}

\begin{proof}[Proof of Theorem~\ref{thm:NPhard_JR-k} and \ref{thm:NPhard_EJR-k}]
    Both theorems can be proved with only minor modifications to the last step in the proof of Theorem~\ref{thm:NPhard_JR} and \ref{thm:NPhard_EJR} above:
    \begin{equation*}
        \left(\frac{1}{\bar{n}' + \bar{m}'}\right)^{1-\epsilon} < \left(\frac{1}{1 + \bar{m}'}\right)^{1-\epsilon} = \left(\frac{1}{k_{\text{voting}}}\right)^{1-\epsilon}.\qedhere
    \end{equation*}
\end{proof}

\section{Parameterized Complexity of {\tt\textbf{MDJR}}/{\tt\textbf{MDEJR}}} \label{sec: fpt}
The parameterized approach is often used to address problems that are hard to solve in their general form but become more tractable or have improved algorithms when considering specific parameter values. In most scenarios, the committee size $k$ is much smaller than the number of voters. Hence, would it be helpful if we fixed the parameter $k$?

In Appendix~\ref{append:parameterized}, we provide a brief introduction to the parameterized complexity theory which includes basic definitions that are necessary for the results in the remaining part of our paper. The readers familiar with parameterized complexity can move on to Section~\ref{sect:W2hard}.

\subsection{W[2]-Hardness with Parameter $k$.}
\label{sect:W2hard}
We first show that both \MDJR and \MDEJR are intractable when the committee size $k$ is specified as a parameter.

\begin{theorem}\label{thm:W2-JR}
    \MDJR is W[2]-hard parameterized by $k$.
\end{theorem}
\begin{proof}
    We present a reduction from the \emph{set cover problem}, a canonical W[2]-complete problem~\cite{downey1995fixed}.
    Given a set cover instance $(\mathcal{U}=\{1,\ldots,n\},\mathcal{S}=\{S_1,\ldots,S_m\},k)$, we construct an \MDJR instance $(N, C, \mathbf{A}, k')$ as follows.

    The set of voters $N$ is given by $N=N_1\cup N_2$ where $N_1=\{1,\ldots,n\}$ corresponds to $\mathcal{U}$ in the set cover instance and $N_2=\{n+1,\ldots,2n\}$ is the set of $n$ additional voters.
    The set of candidates $C=\{c_1,\ldots,c_m\}$ corresponds to $\mathcal{S}$.
    The profile $\mathbf{A}$ is defined as follows:
    \begin{itemize}
        \item for voters in $N_1$, a voter $i$ approves a candidate $c_j$ if and only if $i\in S_j$ in the set cover instance;
        \item for every voter in $N_2$, they approve all candidates.
    \end{itemize}
    Finally, the committee size $k'$ is set to $k'=k$.
    We assume $k\geq 2$ without loss of generality.

    If the set cover instance is a yes instance, by selecting the $k$ candidates corresponding to the $k$ subsets that cover $\mathcal{U}$, every voter approves at least one candidate in the winner committee. The JR degree reaches the maximum possible value $|N|/k'$.

    If the set cover instance is a no instance, for every winner committee with $k'=k$ candidates, there exists a voter in $N_1$ that approves no candidate in the winner committee. On the other hand, every voter in $N_1$ is in at least one cohesive group, as this voter and the $n$ voters in $N_2$, with a total of $n+1>|N|/k'$ voters, approve one common candidate.
    This implies the JR degree cannot reach the maximum possible value $|N|/k'$.
\end{proof}


The hardness of \MDEJR is proved differently.
In the previous section, we prove the hardness of \MDEJR in the same way as \MDJR by making sure only $1$-cohesive groups exist.
However, we fail to make this technique work for proving the following theorem.
Instead, the hard instance constructed here contains $\ell$-cohesive groups for $\ell>1$.

\begin{theorem}\label{thm:W2-EJR}
    \MDEJR is W[2]-hard parameterized by $k$.
\end{theorem}
\begin{proof}
We again present a reduction from the W[2]-complete problem, the \emph{set cover problem}.
We make the following assumptions on the set cover instance $(\bar{U}=\{1,\ldots,\bar{n}\},\bar{S}={S_1,\ldots,S_{\bar{m}}},\bar{k})$ without loss of generality.
\begin{itemize}
    \item $\bar{k}<\bar{m}$. Otherwise, the set cover instance is trivial.
    \item $\bar{n}$ is a multiple of $9$. To achieve this, we can just add dummy elements that are covered by all the subsets.
    \item For any two subsets $S_i,S_j\in \bar{S}$, we have $|S_i\cap S_j|\leq\frac1{75}|\bar{U}|=\bar{n}$. To achieve this, for the original instance, we can create $74\bar{n}$ additional elements, create a subset that contains these elements, and increase $\bar{k}$ by $1$.
\end{itemize}
Given a set cover instance $(\bar{U}=\{1,\ldots,\bar{n}\},\bar{S}={S_1,\ldots,S_{\bar{m}}},\bar{k})$ satisfying the above three assumptions, we construct an \MDEJR instance as follows.
Let $\bar{n}'=\frac43\bar{n}$, and notice that $\bar{n}'$ is a multiple of $3$.

The \MDEJR instance has $n=\bar{n}'\cdot(\bar{k}+3)$ voters which form four groups: $U$, $U'$, $V$, and $W$.
It has $m=\bar{m}+4$ candidates, $c_1, \ldots, c_{\bar{m}},$ $c^\ast, d_1, d_2, d_3$. The size of the winner committee $k$ is set to $k=\bar{k}+3$.
Voter group $U$ and candidates $c_1,\ldots,c_{\bar{m}}$ correspond to the set cover instance, where voter $i$ approves candidate $c_j$ if and only if $i\in S_j$ in the set cover instance.
In addition, all voters in $U$ approve $c^\ast$.
Voter group $U'$ consists of $\bar{n}'-\bar{n}=\frac13\bar{n}$ voters, who only approve $c^\ast$.
Voter group $V$ consists of $(k-3)\cdot\frac{n}k$ voters, and each voter in $V$ approves $\{c_1,\ldots,c_{\bar{m}}\}$.
Finally, voters in $W$ are partitioned into $6$ groups $W_1,\ldots,W_6$ such that each group contains $\frac{n}{3k}$ voters.
Candidate $d_1$ is approved by voters in $W_1\cup W_2\cup W_3$, candidate $d_2$ is approved by voters in $W_3\cup W_4\cup W_5$, and candidate $d_3$ is approved by voters in $W_5\cup W_6\cup W_1$.
Since we have $|U|=\bar{n}=\frac{3n}{4k}$ and $|U'|=\frac13\bar{n}=\frac{n}{4k}$, the total number of voters sums up to $n$:
$$|U|+|U'|+|V|+|W|=\frac{3n}{4k}+\frac{n}{4k}+(k-3)\frac nk+6\cdot\frac{n}{3k}=n.$$
Since we have seen that $\bar{n}$ is a multiple of $3$, $\frac{n}{3k}$ is an integer.

We list the following two key observations:
\begin{enumerate}
    \item To make EJR degree positive, we need to select at least $\bar{k}$ candidates from $\{c_1,\ldots,c_{\bar{m}}\}$. To see this, voters in $V$ approve $\bar{m}\geq\bar{k}=k-3$ common candidates, so $V$ is a $(k-3)$-cohesive group since $|V|=(k-3)\cdot\frac nk$.
    If less than $\bar{k}=k-3$ candidates are selected from $\{c_1,\ldots,c_{\bar{m}}\}$, none of the voters in $V$ is represented, and the EJR degree becomes $0$.
    \item If less than $3$ candidates are selected from $\{d_1,d_2,d_3\}$, the EJR degree is at most $\frac{2n}{3k}$.
    To see this, the three sets of voters, $W_1\cup W_2\cup W_3$, $W_3\cup W_4\cup W_5$, and $W_5\cup W_6\cup W_1$, are $1$-cohesive groups with size exactly $\frac nk$.
    Since voters in $W_2$, $W_4$, and $W_6$ have only one approved candidate $d_1$, $d_2$, and $d_3$ respectively, if one of $\{d_1,d_2,d_3\}$ is not selected, voters from one of $W_2$, $W_4$, and $W_6$ are not represented, which implies the EJR degree is at most $2\cdot\frac{n}{3k}$.
\end{enumerate}
By the observations above, to have a winner committee with the EJR degree at least $\frac{2n}{3k}$, we need to select exactly $\bar{k}$ candidates from $c_1,\ldots,c_{\bar{m}}$ and all the $3$ candidates $d_1,d_2,d_3$.
In particular, we have no chance to select $c^\ast$, which induces the $1$-cohesive group $U\cup U'$ with $|U\cup U'|=\frac nk$.

If the set cover instance is a yes instance, we will show that we can achieve an EJR degree of at least $\frac{3n}{4k}$. 
We select those $\bar{k}$ candidates in $\{c_1,\ldots,c_{\bar{m}}\}$ that correspond to the set cover solution, and we additionally select $d_1,d_2$, and $d_3$.
For all the three $1$-cohesive groups within $W$, all voters are represented, and the EJR degree here is $\frac nk>\frac{3n}{4k}$.
For the $1$-cohesive group $U\cap U'$, all voters in $U$ are represented (since the set cover instance is a yes instance), and the EJR degree here is exactly $|U|=\bar{n}=\frac34\bar{n}'=\frac{3n}{4k}$.
Now we reason about those $\ell$-cohesive groups with $\ell>1$.
Notice that, for at least $2\cdot\frac{n}k$ voters to approve at least $2$ common candidates, the $2$ common candidates can only come from $\{c_1,\ldots,c_{\bar{m}}\}$ (since each of $c^\ast,d_1,d_2,d_3$ is approved by only $n/k$ voters).
For any $c_i,c_j\in \{c_1,\ldots,c_{\bar{m}}\}$, by our third assumption on the set cover instance, the voters that approve both $c_i$ and $c_j$ are those in $V$ together with at most $\frac1{75}\bar{n}$ voters in $U$.
Note also that all voters in $V$ approve $\bar{k}=k-3$ candidates in the winner committee and $k-3$ is the maximum value of $\ell$ for the existence of $\ell$-cohesive groups.
Given any $\ell$-cohesive group with $\ell>1$, at most $\frac1{75}\bar{n}$ voters are not represented.
Thus, the EJR degree for those $\ell$-cohesive groups with $\ell>1$ is at least
$$\frac nk-\frac1{75}\cdot\bar{n}=\frac nk-\frac1{75}\cdot\frac{3n}{4k}=\frac{99n}{100k}>\frac{3n}{4k}.$$
In conclusion, we have shown that the EJR degree for the committee constructed above is $\frac{3n}{4k}$.

If the set cover instance is a no instance, we will show that the EJR degree for any committee is strictly less than $\frac{3n}{4k}$.
Suppose for the sake of contradiction that this is not the case, and we have a winner committee with an EJR degree of at least $\frac{3n}{4k}$.
Since $\frac{3n}{4k}>\frac{2n}{3k}$, our two observations implies that we have to select $\bar{k}$ candidates from $c_1,\ldots,c_{\bar{m}}$ and we have to select all of $d_1,d_2,d_3$.
Since the set cover instance is a no instance, at least one voter in $U$ has no approved candidate in the winner committee.
Moreover, since $c^\ast$ is not selected, all voters in $U'$ do not have an approved winner candidate.
Therefore, in the $1$-cohesive group $U\cup U'$ of size $n/k$, at least $1+|U'|=1+\frac13\bar{n}=\frac{n}{4k}+1$ voters are not represented.
As a result, the EJR degree is at most $\frac{3n}{4k}-1$, which contradicts our assumption that the EJR degree is at least $\frac{3n}{4k}$.
\end{proof}

\subsection{Fixed-Parameter-Tractability with Parameters $k$ and Maximum (E)JR Degree}
\label{sect:FPTtwo}
We have seen that both \MDJR and \MDEJR are still computationally hard even parameterized by $k$.
Thus, to make the problems tractable, different choices of the parameters or additional parameters are needed.

If we choose the number of candidates $m$ as the parameter, it is easy to verify that both \MDJR and \MDEJR are fixed-parameter-tractable. To see this, we can enumerate all the $\binom{m}{k}$ committees. For each committee, we can compute the (E)JR degree in $O(2^mm^2n)$ time~\cite{aziz2018complexity}. At last, we select the committee that achieves the maximum (E)JR degree. 

Another natural choice for the parameter is the maximum achievable (E)JR degree.
Fortunately, both \MDJR and \MDEJR become tractable if parameterized by both $k$ and the maximum (E)JR degree.
In the next two sections, we use $c_{\text{max}}$ to denote the maximum (E)JR degree.

\subsubsection{Algorithm for \MDJR}
Our starting point is the algorithm GreedyAV (Algorithm~\ref{alg: greedyAV}) proposed by~\citet{DBLP:journals/scw/AzizBCEFW17}.
We show the following property for GreedyAV which is the key for our algorithm.
It states that the algorithm GreedyAV also gives us the optimal JR degree if $n$ is large enough.
 
\begin{proposition} \label{prop: greedy opt}
    Given any instance with the maximum achievable JR degree of $c_{\max}$ and $n>k^2(c_{\max}-1)$, GreedyAV will output a committee achieving JR degree $c_{\max}$.
\end{proposition}
\begin{proof}
    $n>k^2(c_{\max}-1)$ implies $c_{\max}<\frac{n}{k^2}+1$.
    This proposition then follows from Proposition~\ref{prop:MDJR-approximation}.
\end{proof}

Given any instance with $n>k^2(c_{\text{max}}-1)$, GreedyAV can achieve the maximum JR degree. Hence, the remaining case is $n\le k^2(c_{\text{max}}-1)$. Given any instance $(N, C, \mathbf{A}, k)$, each candidate $c_i\in C$ can be seen as a subset of $N$ including all voters that approve $c_i$. Hence, there are at most $2^n$ different types of candidates, implying that every instance corresponds to an equivalent instance with $m\le 2^n$. Therefore, we can decide whether there exists a committee providing JR degree of $c$ by enumerating all the committees when $n\le k^2(c-1)$, which can be computed with running time 
\begin{equation*}
 \binom{m}{k}\le \binom{2^n}{k}\le  \binom{2^{k^2(c-1)}}{k}= f(k,c).
\end{equation*}

\begin{algorithm}[!htb]
\caption{\MDJR Voting Rule}
\label{alg: MDJR}
\KwIn{An instance $\mathcal{I} = (N, C, \mathbf{A}, k))$}
\KwOut{A winning committee $W$ of size $k$}
$W\gets \text{GreedyAV}(\mathcal{I})$

\For{$c: \lceil\frac{n}{k^2}\rceil$ to $\lfloor\frac{n}{k}\rfloor$}{
    Enumerate all those $\binom{m}{k}$ possible committees to see if there is a committee $W^*$ achieving JR degree $c$\;
    \eIf{$W^*$ exists}{
        $W\gets W^*$
    }{
        \Return $W$
    }
}
\end{algorithm}

Our algorithm is described in Algorithm~\ref{alg: MDJR}.
    We prove that Algorithm~\ref{alg: MDJR} outputs the committee that achieves the maximum JR degree and runs in time $f(k,c_{\text{max}})\cdot \text{poly}(m,n)$.
    
    \paragraph{Running Time.} GreedyAV can be computed in polynomial time. It is easy to see that the number of \textbf{for} loop is bounded by $O(\frac{n}{k})$. Then, we only need to justify the running time of deciding the condition of \textbf{if} branch. Since $c>\frac{n}{k^2}$, we have $n<k^2c$. As we mentioned before, we can just enumerate all committees in running time $\binom{2^{k^2c}}{k}$ and decide the JR degree for each committee in polynomial time.

    \paragraph{Correctness.} If \MDJR voting rule executes \textbf{else} branch in the first iteration of \textbf{for} loop, there does not exist a committee with JR degree $c$ for $c>\frac{n}{k^2}$. Hence, we have $c_{\text{max}}\le \frac{n}{k^2}$. From Proposition \ref{prop: greedy opt}, we know that GreedyAV achieves JR degree $c_{\text{max}}$. For the other cases, suppose that the \textbf{else} branch is executed when $c=c_{\text{end}}$. Then there does not exist a committee with JR degree $c_{\text{end}}$. The latest $W$ is the committee that provides JR degree $(c_{\text{end}}-1)$, which achieves the maximum JR degree.

\subsubsection{Algorithm for \MDEJR}

For \MDEJR, we prove the following observation in a similar spirit to Proposition~\ref{prop: greedy opt}, which shows that the local search variant of PAV (Algorithm~\ref{alg:LS-PAV}) can achieve maximum EJR degree if $n$ is sufficiently large and $\lambda$ is sufficiently small.

\begin{proposition} \label{prop: LSPAV_IS_GOOD}
    Given any instance with the maximum achievable EJR degree of $c_{\max}$. For any $\lambda\in[0,\frac n{k^2})$ satisfying $n> k(k+1)(c_{\max}-1)+\lambda k^2$, $\lambda$-LS-PAV will output the committee achieving EJR degree $c_{\max}$.
\end{proposition}

\begin{proof}
    $n> k(k+1)(c_{\max}-1)+\lambda k^2$ implies $c_{\max}<\frac{n}{k(k+1)}-\lambda\frac{k}{k+1}+1$.
    Proposition~\ref{prop:MDEJR-approximation} implies this proposition.
\end{proof}

Since Proposition \ref{prop: LSPAV_IS_GOOD} holds for every initial committee $W$ in $\lambda$-LS-PAV, by considering the initial committee $W$ being the one with the maximum PAV-score and $\lambda=0$ (in which case the while-loop is never executed), we have the following corollary.
In contrast to our counterexample in Appendix~\ref{exxample : PAV fail mdejr}, this corollary shows the success of PAV for large enough $n$, which may be of independent interest.
\begin{corollary}
    Given any instance with the maximum achievable EJR degree of $c_{\max}$. If $n> k(k+1)(c_{\max}-1)$, PAV has an EJR degree of $c_{\max}$.
\end{corollary}   

By setting $\lambda=\frac{n}{k(k+1)}$ in Proposition~\ref{prop: LSPAV_IS_GOOD}, we have the following corollary, which is crucial for our algorithm.

\begin{corollary} \label{Coro : special_case}
    Given any instance $(N, C, \mathbf{A}, k)$ with maximum achievable EJR degree $c_{\max}$. If $n> k(k+1)^2(c_{\max}-1)$, $\frac{n}{k(k+1)}$-LS-PAV outputs a committee with EJR degree $c_{\max}$.
\end{corollary}


Based on Corollary \ref{Coro : special_case}, we can use a similar way to \MDJR to design our algorithm. In particular, we can decide whether there exists a committee providing EJR degree $c$ by enumerating all the committees when $n\le k(k+1)^2(c_{\max}-1)$, which can be computed in running time $f(k,c)$. Our algorithm is presented in Algorithm~\ref{alg: MDEJR}, which achieves the maximum EJR degree and runs in time $f(k,c_{\max})\cdot \text{poly}(m,n)$. The correctness of our algorithm follows from arguments similar to \MDJR.
The time complexity analysis is also similar, except that we need to show $\frac{n}{k(k+1)}$-LS-PAV runs in polynomial time.
To see this, each swap operation can be executed in polynomial time with a minimum score increment by $\lambda=\frac{n}{k(k+1)}$, and the maximum PAV-score is $n \cdot (1+\frac{1}{2}+\cdots+\frac{1}{k})=O(n\ln k)$.
Thus, the number of while-loop executions is bounded by $O(k^2\ln k)$.

\begin{algorithm}[tb]
\caption{\MDEJR Voting Rule}
\label{alg: MDEJR}
\KwIn{An instance $\mathcal{I} = (N, C, \mathbf{A}, k)$}
\KwOut{A winning committee $W$ of size $k$}
$W\gets \frac{n}{k(k+1)}\text{-LS-PAV}(\mathcal{I})$

\For{$c: \lceil\frac{n}{k(k+1)^2}\rceil$ to $\lfloor\frac{n}{k}\rfloor$}{
    Enumerate all those $\binom{m}{k}$ possible committees to see if there is a committee $W^*$ achieving EJR degree $c$\;
    \eIf{$W^*$ exists}{
        $W\gets W^*$
    }{
        \Return $W$
    }
}
\end{algorithm}

    


\section{Conclusion and Future Work} \label{sec: con}
We initialize the study of the (E)JR degree and study its computational complexity and approximability. The (E)JR degree describes (E)JR from a quantitative perspective, which can help us better compare 
different committees. 
When explaining (E)JR from a stability perspective, i.e., for any $\ell$-cohesive group, if this group deviates and constructs $\ell$ winners, then at least one member does not want to deviate, as the current satisfaction is already $\ell$, which is the maximum satisfaction possible with $\ell$ winners. However, in reality, if only one person is represented, they can easily be persuaded to deviate by the rest of the group. Hence, if a committee provides a larger (E)JR degree, the possibility of deviation will be reduced. Moreover, we give complete pictures of both optimization problems, from a general NP-hardness with almost tight inapproximability to a parameterized complexity analysis with some natural parameters. 

Many potential further works can be explored. For example, one can explore whether the negative results can be circumvented by considering restricted domains of preferences \citep{DBLP:conf/ijcai/ElkindL15,DBLP:conf/ijcai/Yang19a}. Another direction is to consider some other quantitative measurements with respect to (E)JR, e.g., using the ratio instead of the number. 
As we have discussed earlier, our definition of EJR degree using numbers is more aligned with the original definition of EJR, whereas the definition using ratios/fractions, by prioritizing more on $\ell$-cohesive groups with large $\ell$, is more aligned with the notion of Individual Representation~\cite{DBLP:conf/sigecom/Brill023} discussed in Sect.~\ref{sec:relatedwork}. 
We believe both the ``number version'' and the ``ratio version'' are worth studying. Which choice is better depends on the specific applications.


\section*{Acknowledgments}
The research of Biaoshuai Tao was supported by the National Natural Science Foundation of China (No. 62472271) and the Key Laboratory of Interdisciplinary Research of Computation and Economics (Shanghai University of Finance and Economics), Ministry of Education.

The authors gratefully appreciate the valuable comments from IJCAI'24, AAAI'25, and AAMAS'25 reviewers, which are very helpful in improving this paper. 
We especially appreciate that one of the AAMAS'25 reviewers points out the $\frac1k$-approximation of the greedy algorithm, which leads to Theorem~\ref{thm:MDJR-approximation} and \ref{thm:MDEJR-approximation} in this paper.



\bibliographystyle{plainnat}
\bibliography{reference}

\begin{thebibliography}{37}
\providecommand{\natexlab}[1]{#1}
\providecommand{\url}[1]{\texttt{#1}}
\expandafter\ifx\csname urlstyle\endcsname\relax
  \providecommand{\doi}[1]{doi: #1}\else
  \providecommand{\doi}{doi: \begingroup \urlstyle{rm}\Url}\fi

\bibitem[Arrow(2012)]{arrow2012social}
Kenneth~J Arrow.
\newblock \emph{Social choice and individual values}, volume~12.
\newblock Yale university press, 2012.

\bibitem[Aziz et~al.(2017)Aziz, Brill, Conitzer, Elkind, Freeman, and Walsh]{DBLP:journals/scw/AzizBCEFW17}
Haris Aziz, Markus Brill, Vincent Conitzer, Edith Elkind, Rupert Freeman, and Toby Walsh.
\newblock Justified representation in approval-based committee voting.
\newblock \emph{Soc. Choice Welf.}, 48\penalty0 (2):\penalty0 461--485, 2017.
\newblock \doi{10.1007/S00355-016-1019-3}.
\newblock URL \url{https://doi.org/10.1007/s00355-016-1019-3}.

\bibitem[Aziz et~al.(2018{\natexlab{a}})Aziz, Elkind, Huang, Lackner, Fern{\'{a}}ndez, and Skowron]{aziz2018complexity}
Haris Aziz, Edith Elkind, Shenwei Huang, Martin Lackner, Luis~S{\'{a}}nchez Fern{\'{a}}ndez, and Piotr Skowron.
\newblock On the complexity of extended and proportional justified representation.
\newblock In Sheila~A. McIlraith and Kilian~Q. Weinberger, editors, \emph{Proceedings of the Thirty-Second {AAAI} Conference on Artificial Intelligence, (AAAI-18), the 30th innovative Applications of Artificial Intelligence (IAAI-18), and the 8th {AAAI} Symposium on Educational Advances in Artificial Intelligence (EAAI-18), New Orleans, Louisiana, USA, February 2-7, 2018}, pages 902--909. {AAAI} Press, 2018{\natexlab{a}}.
\newblock \doi{10.1609/AAAI.V32I1.11478}.
\newblock URL \url{https://doi.org/10.1609/aaai.v32i1.11478}.

\bibitem[Aziz et~al.(2018{\natexlab{b}})Aziz, Lee, and Talmon]{DBLP:conf/atal/0001LT18}
Haris Aziz, Barton~E. Lee, and Nimrod Talmon.
\newblock Proportionally representative participatory budgeting: Axioms and algorithms.
\newblock In Elisabeth Andr{\'{e}}, Sven Koenig, Mehdi Dastani, and Gita Sukthankar, editors, \emph{Proceedings of the 17th International Conference on Autonomous Agents and MultiAgent Systems, {AAMAS} 2018, Stockholm, Sweden, July 10-15, 2018}, pages 23--31. International Foundation for Autonomous Agents and Multiagent Systems Richland, SC, {USA} / {ACM}, 2018{\natexlab{b}}.
\newblock URL \url{http://dl.acm.org/citation.cfm?id=3237394}.

\bibitem[Behrens et~al.(2014)Behrens, Kistner, Nitsche, and Swierczek]{behrens2014principles}
Jan Behrens, Axel Kistner, Andreas Nitsche, and Bj{\"o}rn Swierczek.
\newblock \emph{The principles of LiquidFeedback}.
\newblock Interacktive Demokratie, 2014.

\bibitem[Bei et~al.(2022)Bei, Lu, and Suksompong]{DBLP:conf/aaai/BeiLS22}
Xiaohui Bei, Xinhang Lu, and Warut Suksompong.
\newblock Truthful cake sharing.
\newblock In \emph{Thirty-Sixth {AAAI} Conference on Artificial Intelligence, {AAAI} 2022, Thirty-Fourth Conference on Innovative Applications of Artificial Intelligence, {IAAI} 2022, The Twelveth Symposium on Educational Advances in Artificial Intelligence, {EAAI} 2022 Virtual Event, February 22 - March 1, 2022}, pages 4809--4817. {AAAI} Press, 2022.
\newblock \doi{10.1609/AAAI.V36I5.20408}.
\newblock URL \url{https://doi.org/10.1609/aaai.v36i5.20408}.

\bibitem[Brill and Peters(2023)]{DBLP:conf/sigecom/Brill023}
Markus Brill and Jannik Peters.
\newblock Robust and verifiable proportionality axioms for multiwinner voting.
\newblock In Kevin Leyton{-}Brown, Jason~D. Hartline, and Larry Samuelson, editors, \emph{Proceedings of the 24th {ACM} Conference on Economics and Computation, {EC} 2023, London, United Kingdom, July 9-12, 2023}, page 301. {ACM}, 2023.
\newblock \doi{10.1145/3580507.3597785}.
\newblock URL \url{https://doi.org/10.1145/3580507.3597785}.

\bibitem[Brill et~al.(2022)Brill, Israel, Micha, and Peters]{DBLP:conf/aaai/BrillIM022}
Markus Brill, Jonas Israel, Evi Micha, and Jannik Peters.
\newblock Individual representation in approval-based committee voting.
\newblock In \emph{Thirty-Sixth {AAAI} Conference on Artificial Intelligence, {AAAI} 2022, Thirty-Fourth Conference on Innovative Applications of Artificial Intelligence, {IAAI} 2022, The Twelveth Symposium on Educational Advances in Artificial Intelligence, {EAAI} 2022 Virtual Event, February 22 - March 1, 2022}, pages 4892--4899. {AAAI} Press, 2022.

\bibitem[Cabannes(2004)]{cabannes2004participatory}
Yves Cabannes.
\newblock Participatory budgeting: a significant contribution to participatory democracy.
\newblock \emph{Environment and urbanization}, 16\penalty0 (1):\penalty0 27--46, 2004.

\bibitem[Cevallos and Stewart(2021)]{DBLP:conf/aft/CevallosS21}
Alfonso Cevallos and Alistair Stewart.
\newblock A verifiably secure and proportional committee election rule.
\newblock In Foteini Baldimtsi and Tim Roughgarden, editors, \emph{{AFT} '21: 3rd {ACM} Conference on Advances in Financial Technologies, Arlington, Virginia, USA, September 26 - 28, 2021}, pages 29--42. {ACM}, 2021.
\newblock \doi{10.1145/3479722.3480988}.
\newblock URL \url{https://doi.org/10.1145/3479722.3480988}.

\bibitem[Downey and Fellows(1995)]{downey1995fixed}
Rod~G Downey and Michael~R Fellows.
\newblock Fixed-parameter tractability and completeness i: Basic results.
\newblock \emph{SIAM Journal on computing}, 24\penalty0 (4):\penalty0 873--921, 1995.

\bibitem[Dwork et~al.(2001)Dwork, Kumar, Naor, and Sivakumar]{DBLP:conf/www/DworkKNS01}
Cynthia Dwork, Ravi Kumar, Moni Naor, and D.~Sivakumar.
\newblock Rank aggregation methods for the web.
\newblock In Vincent~Y. Shen, Nobuo Saito, Michael~R. Lyu, and Mary~Ellen Zurko, editors, \emph{Proceedings of the Tenth International World Wide Web Conference, {WWW} 10, Hong Kong, China, May 1-5, 2001}, pages 613--622. {ACM}, 2001.
\newblock \doi{10.1145/371920.372165}.
\newblock URL \url{https://doi.org/10.1145/371920.372165}.

\bibitem[Elkind and Lackner(2015)]{DBLP:conf/ijcai/ElkindL15}
Edith Elkind and Martin Lackner.
\newblock Structure in dichotomous preferences.
\newblock In Qiang Yang and Michael~J. Wooldridge, editors, \emph{Proceedings of the Twenty-Fourth International Joint Conference on Artificial Intelligence, {IJCAI} 2015, Buenos Aires, Argentina, July 25-31, 2015}, pages 2019--2025. {AAAI} Press, 2015.
\newblock URL \url{http://ijcai.org/Abstract/15/286}.

\bibitem[Elkind et~al.(2022{\natexlab{a}})Elkind, Faliszewski, Igarashi, Manurangsi, Schmidt{-}Kraepelin, and Suksompong]{DBLP:conf/aaai/ElkindFIMSS22}
Edith Elkind, Piotr Faliszewski, Ayumi Igarashi, Pasin Manurangsi, Ulrike Schmidt{-}Kraepelin, and Warut Suksompong.
\newblock The price of justified representation.
\newblock In \emph{Thirty-Sixth {AAAI} Conference on Artificial Intelligence, {AAAI} 2022, Thirty-Fourth Conference on Innovative Applications of Artificial Intelligence, {IAAI} 2022, The Twelveth Symposium on Educational Advances in Artificial Intelligence, {EAAI} 2022 Virtual Event, February 22 - March 1, 2022}, pages 4983--4990. {AAAI} Press, 2022{\natexlab{a}}.
\newblock \doi{10.1609/AAAI.V36I5.20429}.
\newblock URL \url{https://doi.org/10.1609/aaai.v36i5.20429}.

\bibitem[Elkind et~al.(2022{\natexlab{b}})Elkind, Li, and Zhou]{DBLP:conf/atal/ElkindLZ22}
Edith Elkind, Minming Li, and Houyu Zhou.
\newblock Facility location with approval preferences: Strategyproofness and fairness.
\newblock In Piotr Faliszewski, Viviana Mascardi, Catherine Pelachaud, and Matthew~E. Taylor, editors, \emph{21st International Conference on Autonomous Agents and Multiagent Systems, {AAMAS} 2022, Auckland, New Zealand, May 9-13, 2022}, pages 391--399. International Foundation for Autonomous Agents and Multiagent Systems {(IFAAMAS)}, 2022{\natexlab{b}}.
\newblock \doi{10.5555/3535850.3535895}.
\newblock URL \url{https://www.ifaamas.org/Proceedings/aamas2022/pdfs/p391.pdf}.

\bibitem[Fern{\'{a}}ndez and Fisteus(2019)]{DBLP:conf/atal/FernandezF19}
Luis~S{\'{a}}nchez Fern{\'{a}}ndez and Jes{\'{u}}s~A. Fisteus.
\newblock Monotonicity axioms in approval-based multi-winner voting rules.
\newblock In Edith Elkind, Manuela Veloso, Noa Agmon, and Matthew~E. Taylor, editors, \emph{Proceedings of the 18th International Conference on Autonomous Agents and MultiAgent Systems, {AAMAS} '19, Montreal, QC, Canada, May 13-17, 2019}, pages 485--493. International Foundation for Autonomous Agents and Multiagent Systems, 2019.
\newblock URL \url{http://dl.acm.org/citation.cfm?id=3331731}.

\bibitem[Fern{\'{a}}ndez et~al.(2017)Fern{\'{a}}ndez, Elkind, Lackner, Garc{\'{\i}}a, Arias{-}Fisteus, Basanta{-}Val, and Skowron]{DBLP:conf/aaai/FernandezELGABS17}
Luis~S{\'{a}}nchez Fern{\'{a}}ndez, Edith Elkind, Martin Lackner, Norberto~Fern{\'{a}}ndez Garc{\'{\i}}a, Jes{\'{u}}s Arias{-}Fisteus, Pablo Basanta{-}Val, and Piotr Skowron.
\newblock Proportional justified representation.
\newblock In Satinder Singh and Shaul Markovitch, editors, \emph{Proceedings of the Thirty-First {AAAI} Conference on Artificial Intelligence, February 4-9, 2017, San Francisco, California, {USA}}, pages 670--676. {AAAI} Press, 2017.
\newblock \doi{10.1609/AAAI.V31I1.10611}.
\newblock URL \url{https://doi.org/10.1609/aaai.v31i1.10611}.

\bibitem[Goel et~al.(2016)Goel, Krishnaswamy, Sakshuwong, and Aitamurto]{goel2016knapsack}
Ashish Goel, Anilesh~K Krishnaswamy, Sukolsak Sakshuwong, and Tanja Aitamurto.
\newblock Knapsack voting: Voting mechanisms for participatory budgeting.
\newblock \emph{Unpublished manuscript}, 18, 2016.

\bibitem[Israel and Brill(2021)]{DBLP:conf/ijcai/IsraelB21}
Jonas Israel and Markus Brill.
\newblock Dynamic proportional rankings.
\newblock In Zhi{-}Hua Zhou, editor, \emph{Proceedings of the Thirtieth International Joint Conference on Artificial Intelligence, {IJCAI} 2021, Virtual Event / Montreal, Canada, 19-27 August 2021}, pages 261--267. ijcai.org, 2021.
\newblock \doi{10.24963/IJCAI.2021/37}.
\newblock URL \url{https://doi.org/10.24963/ijcai.2021/37}.

\bibitem[Janson(2016)]{janson2016phragmen}
Svante Janson.
\newblock Phragm{\'e}n's and thiele's election methods.
\newblock \emph{ArXiv e-prints}, pages arXiv--1611, 2016.

\bibitem[Lackner and Skowron(2020)]{DBLP:journals/ai/LacknerS20}
Martin Lackner and Piotr Skowron.
\newblock Utilitarian welfare and representation guarantees of approval-based multiwinner rules.
\newblock \emph{Artif. Intell.}, 288:\penalty0 103366, 2020.
\newblock \doi{10.1016/J.ARTINT.2020.103366}.
\newblock URL \url{https://doi.org/10.1016/j.artint.2020.103366}.

\bibitem[Lackner and Skowron(2021)]{DBLP:journals/jet/LacknerS21}
Martin Lackner and Piotr Skowron.
\newblock Consistent approval-based multi-winner rules.
\newblock \emph{J. Econ. Theory}, 192:\penalty0 105173, 2021.
\newblock \doi{10.1016/J.JET.2020.105173}.
\newblock URL \url{https://doi.org/10.1016/j.jet.2020.105173}.

\bibitem[Lackner and Skowron(2023)]{DBLP:series/sbis/LacknerS23}
Martin Lackner and Piotr Skowron.
\newblock \emph{Multi-Winner Voting with Approval Preferences - Artificial Intelligence, Multiagent Systems, and Cognitive Robotics}.
\newblock Springer Briefs in Intelligent Systems. Springer, 2023.
\newblock ISBN 978-3-031-09015-8.
\newblock \doi{10.1007/978-3-031-09016-5}.
\newblock URL \url{https://doi.org/10.1007/978-3-031-09016-5}.

\bibitem[Lu and Boutilier(2011)]{DBLP:conf/ijcai/LuB11}
Tyler Lu and Craig Boutilier.
\newblock Budgeted social choice: From consensus to personalized decision making.
\newblock In Toby Walsh, editor, \emph{{IJCAI} 2011, Proceedings of the 22nd International Joint Conference on Artificial Intelligence, Barcelona, Catalonia, Spain, July 16-22, 2011}, pages 280--286. {IJCAI/AAAI}, 2011.
\newblock \doi{10.5591/978-1-57735-516-8/IJCAI11-057}.
\newblock URL \url{https://doi.org/10.5591/978-1-57735-516-8/IJCAI11-057}.

\bibitem[Lu and Boutilier(2015)]{DBLP:conf/aaai/LuB15}
Tyler Lu and Craig Boutilier.
\newblock Value-directed compression of large-scale assignment problems.
\newblock In Blai Bonet and Sven Koenig, editors, \emph{Proceedings of the Twenty-Ninth {AAAI} Conference on Artificial Intelligence, January 25-30, 2015, Austin, Texas, {USA}}, pages 1182--1190. {AAAI} Press, 2015.
\newblock \doi{10.1609/AAAI.V29I1.9364}.
\newblock URL \url{https://doi.org/10.1609/aaai.v29i1.9364}.

\bibitem[Lu et~al.(2023)Lu, Peters, Aziz, Bei, and Suksompong]{DBLP:conf/aaai/Lu00BS23}
Xinhang Lu, Jannik Peters, Haris Aziz, Xiaohui Bei, and Warut Suksompong.
\newblock Approval-based voting with mixed goods.
\newblock In Brian Williams, Yiling Chen, and Jennifer Neville, editors, \emph{Thirty-Seventh {AAAI} Conference on Artificial Intelligence, {AAAI} 2023, Thirty-Fifth Conference on Innovative Applications of Artificial Intelligence, {IAAI} 2023, Thirteenth Symposium on Educational Advances in Artificial Intelligence, {EAAI} 2023, Washington, DC, USA, February 7-14, 2023}, pages 5781--5788. {AAAI} Press, 2023.
\newblock \doi{10.1609/AAAI.V37I5.25717}.
\newblock URL \url{https://doi.org/10.1609/aaai.v37i5.25717}.

\bibitem[May(1952)]{may1952set}
Kenneth~O May.
\newblock A set of independent necessary and sufficient conditions for simple majority decision.
\newblock \emph{Econometrica: Journal of the Econometric Society}, pages 680--684, 1952.

\bibitem[Monroe(1995)]{monroe1995fully}
Burt~L Monroe.
\newblock Fully proportional representation.
\newblock \emph{American Political Science Review}, 89\penalty0 (4):\penalty0 925--940, 1995.

\bibitem[Moulin(1991)]{moulin1991axioms}
Herv{\'e} Moulin.
\newblock \emph{Axioms of cooperative decision making}.
\newblock Number~15. Cambridge university press, 1991.

\bibitem[Peters(2018)]{DBLP:conf/atal/Peters18}
Dominik Peters.
\newblock Proportionality and strategyproofness in multiwinner elections.
\newblock In Elisabeth Andr{\'{e}}, Sven Koenig, Mehdi Dastani, and Gita Sukthankar, editors, \emph{Proceedings of the 17th International Conference on Autonomous Agents and MultiAgent Systems, {AAMAS} 2018, Stockholm, Sweden, July 10-15, 2018}, pages 1549--1557. International Foundation for Autonomous Agents and Multiagent Systems Richland, SC, {USA} / {ACM}, 2018.
\newblock URL \url{http://dl.acm.org/citation.cfm?id=3237931}.

\bibitem[Peters and Skowron(2020)]{DBLP:conf/sigecom/PetersS20}
Dominik Peters and Piotr Skowron.
\newblock Proportionality and the limits of welfarism.
\newblock In P{\'{e}}ter Bir{\'{o}}, Jason~D. Hartline, Michael Ostrovsky, and Ariel~D. Procaccia, editors, \emph{{EC} '20: The 21st {ACM} Conference on Economics and Computation, Virtual Event, Hungary, July 13-17, 2020}, pages 793--794. {ACM}, 2020.
\newblock \doi{10.1145/3391403.3399465}.
\newblock URL \url{https://doi.org/10.1145/3391403.3399465}.

\bibitem[Peters et~al.(2021)Peters, Pierczynski, and Skowron]{DBLP:conf/nips/PierczynskiSP21}
Dominik Peters, Grzegorz Pierczynski, and Piotr Skowron.
\newblock Proportional participatory budgeting with additive utilities.
\newblock In Marc'Aurelio Ranzato, Alina Beygelzimer, Yann~N. Dauphin, Percy Liang, and Jennifer~Wortman Vaughan, editors, \emph{Advances in Neural Information Processing Systems 34: Annual Conference on Neural Information Processing Systems 2021, NeurIPS 2021, December 6-14, 2021, virtual}, pages 12726--12737, 2021.
\newblock URL \url{https://proceedings.neurips.cc/paper/2021/hash/69f8ea31de0c00502b2ae571fbab1f95-Abstract.html}.

\bibitem[Skowron(2021)]{DBLP:conf/sigecom/Skowron21}
Piotr Skowron.
\newblock Proportionality degree of multiwinner rules.
\newblock In P{\'{e}}ter Bir{\'{o}}, Shuchi Chawla, and Federico Echenique, editors, \emph{{EC} '21: The 22nd {ACM} Conference on Economics and Computation, Budapest, Hungary, July 18-23, 2021}, pages 820--840. {ACM}, 2021.
\newblock \doi{10.1145/3465456.3467641}.
\newblock URL \url{https://doi.org/10.1145/3465456.3467641}.

\bibitem[Skowron et~al.(2017)Skowron, Lackner, Brill, Peters, and Elkind]{DBLP:conf/ijcai/SkowronLBPE17}
Piotr Skowron, Martin Lackner, Markus Brill, Dominik Peters, and Edith Elkind.
\newblock Proportional rankings.
\newblock In Carles Sierra, editor, \emph{Proceedings of the Twenty-Sixth International Joint Conference on Artificial Intelligence, {IJCAI} 2017, Melbourne, Australia, August 19-25, 2017}, pages 409--415. ijcai.org, 2017.
\newblock \doi{10.24963/IJCAI.2017/58}.
\newblock URL \url{https://doi.org/10.24963/ijcai.2017/58}.

\bibitem[Thiele(1895)]{Theiele1895}
T.~N. Thiele.
\newblock Om flerfoldsvalg.
\newblock \emph{Oversigt over det Kongelige Danske Videnskabernes Selskabs}, pages 415--441, 1895.

\bibitem[Tovey(1984)]{tovey1984simplified}
Craig~A Tovey.
\newblock A simplified np-complete satisfiability problem.
\newblock \emph{Discrete applied mathematics}, 8\penalty0 (1):\penalty0 85--89, 1984.

\bibitem[Yang(2019)]{DBLP:conf/ijcai/Yang19a}
Yongjie Yang.
\newblock On the tree representations of dichotomous preferences.
\newblock In Sarit Kraus, editor, \emph{Proceedings of the Twenty-Eighth International Joint Conference on Artificial Intelligence, {IJCAI} 2019, Macao, China, August 10-16, 2019}, pages 644--650. ijcai.org, 2019.
\newblock \doi{10.24963/IJCAI.2019/91}.
\newblock URL \url{https://doi.org/10.24963/ijcai.2019/91}.

\end{thebibliography}

\clearpage
\appendix

\section{Parameterized Complexity Basics}
\label{append:parameterized}
A \emph{parameterized problem} is a language $L\subseteq\Sigma^\ast\times\mathbb{Z}^+$ (where $\Sigma$ is the alphabet set) and $p$ in an instance $(x,p)\in\Sigma^\ast\times\mathbb{Z}^+$ is called the \emph{parameter}.
A parameterized problem is \emph{fixed-parameter tractable} (FPT), or a parameterized problem is in the complexity class FPT, if there exists an algorithm that decides the problem with running time bounded by $f(p)\cdot\text{poly}(|(x,p)|)$, where $(x,p)$ is the instance with length $|(x,p)|$, poly$(\cdot)$ is a polynomial function, and $f(\cdot)$ is a computable function.
The definition of parameterized problems and FPT can be straightforwardly generalized to optimization problems.

As introduced above, the instance of a parameterized problem consists of two parts: the ``original instance'' $x$ and the parameter $p$.
In our case of \MDJR, $x$ corresponds to a \MDJR instance $(N, C, \mathbf{A}, k)$, and $p$ is a parameter that we need to specify.
For example, in Section~\ref{sect:W2hard}, we select the winner committee size $k$ as the parameter.
If there is an algorithm that solves \MDJR problem with running time $f(k)\cdot\text{poly}(m,n)$, then we will say that \MDJR is fixed-parameter tractable parameterized by $k$.
(However, as we will see later, this is unlikely the case.)
Here, $f(k)$ can be any computable functions such as $f(k)=2^k$ or $f(k)=2^{2^{k}\cdot k!}$.
Specifically, the running time is not required to be polynomial in terms of the parameter, as long as the (possibly) super-polynomial factor with the parameter in the time complexity can be separated from the polynomial function of the length of the original instance.
Naturally, for the algorithm to be useful in practice, we should choose parameters that are not large in practical scenarios, such as the size of the committee (which is usually significantly smaller than the number of voters).

It is also straightforward to generalize the parameterized problems with more than one parameter.
For example, with two parameters $p_1$ and $p_2$, the problem is in FPT if there exists an algorithm with running time bounded by $f(p_1,p_2)$ times a polynomial function of the input length, where $f(\cdot,\cdot)$ now becomes a computable function of two variables.
In Section~\ref{sect:FPTtwo}, we consider the choice of two parameters: the committee size $k$ and the maximum achievable JR degree $c_{\text{max}}$, and we will show \MDJR is FPT for this choice.

Downey and Fellows define the class of parameterized problems, the \emph{W-hierarchy}, in an attempt to capture hard problems.
Problems that are complete in the classes of W-hierarchy (such as W[1], W[2], etc) are believed to be fixed-parameter intractable.
For the purpose of this paper, we omit the detailed definition of the W-hierarchy, and we refer the readers to Reference~\cite{downey1995fixed} for details.

To show a problem $L$ is W[$t$]-hard, we can give a \emph{parameterized reduction} from a known W[$t$]-hard problem $\bar{L}$.
We say that $\bar{L}$ parameterized reduce to $L$ if, given an $\bar{L}$ instance $(x,p)$, we can construct an $L$ instance $(x',p')$ such that
\begin{enumerate}
    \item $(x,p)$ is a yes instance if and only if $(x',p')$ is a yes instance;
    \item $p'\leq g(p)$ for some computable function $g$;
    \item the construction can be done in $f(k)\cdot\text{poly}(|x|)$ time for some computable function $f$.
\end{enumerate}
It is straightforward to verify that, if $\bar{L}$ parameterized reduces to $L$, then $L$ being fixed-parameter tractable implies that $\bar{L}$ is fixed-parameter tractable.

The notion of W[$t$]-hardness can be extended to optimization problems, just as we can use the word ``NP-hardness'' to describe an optimization problem.
Correspondingly, taking the maximization problem for example, description 1 above is changed such that a yes instance $(x,p)$ is mapped to instances $(x',p')$ with the objective value of at least a certain threshold $\theta$ while a no instance $(x,p)$ is mapped to instances with objective values less than $\theta$.

\section{Proof of Proposition~\ref{prop:difference}}
\label{append:difference}

Fix $\gamma>0$. We construct the following ballot instance.
Let $P$ be a large integer whose value is decided later, and let $L=2P^2-2P$.
The instance consists of $n=P^2(L+2)-2P-L$ voters and $2(L+2)$ candidates $c_1,\ldots,c_{2(L+2)}$.
The winner committee has size $k=2L$.
The ballot is described in the table below.

\begin{small}
    \begin{tabular}{cc}
\hline
  candidates   &  voters who approve the candidates\\
\hline
  $c_1,c_2$   & $1,2,\ldots,P^2$\\
  $c_3,c_4$   & $P^2,P^2+1,\ldots,2P^2-1$\\
  $c_5,c_6$   & $2P^2-1,2P^2,\ldots,3P^2-2$\\
  $c_7,c_8$ & $3P^2-2,3P^2-1,\ldots,4P^2-3$\\
  $\vdots$ & $\vdots$\\
  $c_{2L+1},c_{2L+2}$ & $Lp^2-(L-1),\ldots,(L+1)P^2-L$\\
  $c_{2L+3},c_{2L+4}$ & $\left\{\begin{array}{l}
       (L+1)P^2-L-P+1,\ldots,n \\
       1,2,\ldots,P
  \end{array}\right.$\\
  \hline
\end{tabular}
\end{small}

We begin by identifying all cohesive groups.
Let $V_i$ be the set of voters who approve the two candidates $c_{2i-1},c_{2i}$.
We have $|V_1|=|V_2|=\cdots=|V_{L+2}|=P^2$: this is easy to see for $V_1,\ldots,V_{L+1}$; for $V_{L+2}$, we have
$$|V_{L+2}|=n-((L+1)P^2-L-P+1)+1+P,$$
which equals to $P^2$ by substituting the value of $n=P^2(L+2)-2P-L$.
In addition, we have $P^2=2\cdot\frac nk$ (to see this equation, substitute the values of $n=P^2(L+2)-2P-L$, $L=2P^2-2P$, and $k=2L$), so all of $V_1,\ldots,V_{L+2}$ are $2$-cohesive groups, and it is easy to see that these are the only $2$-cohesive groups.
Those $1$-cohesive groups are exactly those subsets of each $V_i$ with at least $\frac nk=\frac{P^2}2$ voters.
We have identified all the cohesive groups.

By our construction, we have 
$$|V_i\cap V_j|=\left\{\begin{array}{ll}
    1 & \mbox{if }j=i+1\mbox{ and }j\leq L+1 \\
    P & \mbox{if }i=L+1\mbox{ and }j=L+2\\
    P & \mbox{if }i=L+2\mbox{ and }j=1\\
    0 & \mbox{otherwise}
\end{array}\right..$$

To maximize the EJR degree, it is easy to see that the optimal winner committee can only be $\{c_3,c_4,\ldots,c_{2L},c_{2L+3},c_{2L+4}\}$, in which case only voters in $V_{1}$ and $V_{L+1}$ (respectively) are not fully represented.
The optimal EJR degree is $P+1$.

To maximize the JR degree, we need to select at least one candidate from $\{c_{2i-1},c_{2i}\}$ for each $i=1,\ldots,L+2$, in which case all voters have at least $1$ approved candidate.
The optimal JR degree is the maximum possible value $\frac nk=\frac{P^2}2$.
The remaining $L-2$ candidates can be selected arbitrarily.

To see 1 in the proposition, we have $c^\ast_{\text{JR}}/c^\ast_{\text{EJR}}=\frac{P^2}{2P+2}$, which can be made larger than $\gamma$ by large enough $P$.

To see 2, the described winner committee optimal for the EJR degree has a JR degree of $P+1$, which is less than  $c^\ast_{\text{JR}}/\gamma$ for large enough $P$.

To see 3, for a winner committee that optimizes the JR degree, among the $L+2$ candidate groups $\{c_1,c_2\},\ldots,\{c_{2L+3},c_{2L+4}\}$, there are exactly four of them with only one selected candidate.
By the pigeonhole principle, there is at least one of the $L-1$ candidate groups $\{c_3,c_4\},\{c_5,c_6\},\ldots,\{c_{2L-1},c_{2L}\}$, with one selected candidate.
Since $V_i$ intersects with $\bigcup_{j\neq i}V_j$ by only two elements for each $i=2,\ldots,L$, the EJR degree in this case is at most $2$, which is less than $c^\ast_{\text{E1JR}}/\gamma$ for large enough $P$.

\section{Counterexample for PAV failing {\tt \textbf{MDEJR}}}
\label{exxample : PAV fail mdejr}

    Consider an instance with $m=3p+2$ candidates $c_1,c_2,\dots,c_{3p},d_1,d_2$ and $n=3p+(3p+2)+p+(9p^2-p-1)=(3p+1)^2$ voters. The committee size is set to $k=3p+1$. There are $4$ groups of voters, $D_1,D_2,T$, and $S$, with $|D_1|=3p$, $|D_2|=3p+2$, $|T|=p$, and $|S|=9p^2-p-1$. All voters in $D_1$ approve $d_1$ and $d_2$, and the $i$-th voter in $D_1$ approves $c_i$ additionally. All voters in $D_2$ approve $d_1$ and $d_2$. For those $p$ voters in group $T$, the $i$-th voter approves $c_{3i-2},c_{3i-1},c_{3i}$. All voters in $S$ approve $c_1,c_2,\dots,c_{3p}$.

    Note that any size-$k$ committee covers every voter, so the JR-degree is always $\frac{n}{k}=3p+1$. We only need to care about $\ell$-cohesive groups with $\ell\ge 2$. First, any $\ell$-cohesive group with $\ell\ge 4$ is a subset of $S$ since voters in the other three group have degree less than $4$. There is no $\ell$-cohesive group for $\ell\ge k-2=3p-1$ as $(3p-1)(3p+1)=9p^2-1>|S|$. 
    
    At least $3p-1$ candidates in $c_1,\dots,c_{3p}$ are selected, so any $\ell$-cohesive group with $\ell\ge 2$ has $\ell \frac{n}{k}>\frac{n}{k}3p+1$ voters that approve at least $\ell$ candidates in the committee.

    Any voter in $T$ together with some voters in $S$ forms a $2$-cohesive group or $3$-cohesive group, has at least $3p+1$ voters in which approve at least $3p+1$ candidates in the committee.
    
    In addition, $D_1 \cup D_2$ forms a $2$-cohesive group. Since there is only one voter (in $D_1$) approves some $c_i$ and $d_1$ (or $d_2$), $D_1 \cup D_2$ is the only $\ell$-cohesive group that contains voters in $D_1$ and $D_2$. Thus, the EJR-degree of the committee depends on the number of voters in $D_1 \cup D_2$ that approve at least $2$ candidates in the committee. There are only two kinds of different committees:
    \begin{description}
        \item[Case 1] If we select any $3p-1$ candidates in $\{c_1,\dots,c_{3p}\}$ and $d_1,d_2$ to form a winning committee, all voters in $D_1\cup D_2$ approve at least $2$ winning candidates, namely, $d_1$ and $d_2$. Therefore, the EJR degree reaches its maximum, $3p+1$.
        \item[Case 2] If we select all $3p$ candidates in $\{c_1,\dots,c_{3p}\}$ and one of $d_1,d_2$, the EJR-degree is $3p$. To see this, for the $2$-cohesive group $D_1\cup D_2$, all voters in $D_2$ only approve $1$ candidate in the winner committee, so the number of voters who approve $2$ winner candidates in this $2$-cohesive group is $3p$.
    \end{description}
    Consider the winner committee $W=\{c_1,\dots,c_{3p},d_1\}$ in Case $2$. If we swap $d_2$ and some $c_i$, the PAV-score increases by $\frac{1}{3}(|D_1|-1)$ due to voters in $D_1$, by $\frac{1}{2}|D_2|$ due to voters in $D_2$, while the PAV-score decreases by $\frac{1}{3}$ for voters in $T$ and by $\frac{1}{3p}|S|$ for voters in $S$. Therefore, the change in the PAV-score is
    \begin{equation*}\begin{aligned}
        \Delta(W,d_2,c_i) &= \frac{1}{3}(3p-1)+\frac{1}{2}(3p+2)-\frac{1}{3}-\frac{9p^2-p-1}{3p}\\
        &= -\left(\frac{1}{2}p-\frac{2}{3}-\frac{1}{3p}\right).
    \end{aligned}\end{equation*}
    For all $p\ge 2$, $\frac{1}{2}p-\frac{2}{3}-\frac{1}{3p}>0$, which implies Case $2$ has a larger PAV-score compared with Case $1$. The voting rule PAV outputs the committee in Case $2$, which fails to maximize the EJR degree.

\end{document}